\documentclass[12pt]{article}
\input{Qcircuit}
\usepackage{amsmath,amssymb,amsthm}
\usepackage{graphicx}
\usepackage{mathtools}
\usepackage[margin=1in]{geometry}
\usepackage{fancyhdr}
\usepackage{bbm}
\usepackage{graphicx}
\usepackage{mleftright}
\usepackage{pgf}
\usepackage{tikz}
\usepackage{verbatim}
\usepackage{braket}
\usepackage{enumerate}
\usepackage{hyperref}
\usepackage[capitalise]{cleveref}
\usepackage{pgfplots}
\usetikzlibrary{arrows}
\newtheorem{theorem}{Theorem}[section]

\newtheorem{corollary}[theorem]{Corollary}
\newtheorem{claim}[theorem]{Claim}

\newtheorem{definition}[theorem]{Definition}
\newtheorem{lemma}[theorem]{Lemma}
\newcommand{\E}{\mathbb{E}}
\newcommand{\p}{\mathbb{P}}

\title{Quantum versus Randomized Communication Complexity, with Efficient Players}
\date{}
\author{Uma Girish\thanks{Department of Computer Science, Princeton University. Research supported by the Simons Collaboration on Algorithms and Geometry, by a Simons Investigator Award and by the National Science Foundation grant No. CCF-1714779.} 
\and Ran Raz\thanks{Department of Computer Science, Princeton University. Research supported by the Simons Collaboration on Algorithms and Geometry, by a Simons Investigator Award and by the National Science Foundation grant No. CCF-1714779.} 
\and Avishay Tal\thanks{Department of Electrical Engineering and Computer Sciences, University of California at Berkeley. Part of this work was done when the author was a postdoc at the Department of Computer Science, Stanford University. Partially supported by a Motwani Postdoctoral Fellowship and by NSF grant CCF-1763311.}}

\begin{document}
\maketitle

\begin{abstract}
We study a new type of separations between quantum and classical communication complexity, separations that are obtained using quantum protocols where all parties are {\bf efficient}, in the sense that they can be implemented by small quantum circuits, with oracle access to their inputs.
Our main result qualitatively  matches the strongest known separation between quantum and classical communication complexity~\cite{gavinsky} and is obtained using a quantum protocol where all parties are efficient. More precisely, we
give an explicit partial Boolean function~$f$ over inputs of length $N$, such that:
\begin{enumerate}[(1)]
\item
$f$ can be computed by a simultaneous-message quantum protocol with communication complexity $\mbox{polylog}(N)$
(where at the beginning of the protocol Alice and Bob also have  $\mbox{polylog}(N)$ entangled EPR pairs).
\item
Any classical randomized protocol for $f$, with any number of rounds, has communication complexity at least $\tilde{\Omega}\left(N^{1/4}\right)$.
\item
All parties in the quantum protocol of Item~(1) (Alice, Bob and the referee) can be implemented by quantum circuits of size $\mbox{polylog}(N)$ (where Alice and Bob have oracle access to their inputs).
\end{enumerate}

Items~(1), (2) qualitatively match the strongest known separation between quantum and classical communication complexity, proved by Gavinsky~\cite{gavinsky}. Item~(3) is new. (Our result is incomparable to the one of Gavinsky. While he obtained a quantitatively better lower bound of $\Omega\left(N^{1/2}\right)$ in the classical case, the referee in his quantum protocol is inefficient).

Exponential separations of quantum and classical communication complexity have been studied in numerous previous works, but to the best of our knowledge the efficiency of the parties in the quantum protocol has not been addressed, and in most previous separations the quantum parties seem to be inefficient. The only separations that we know of that have efficient quantum parties are the recent separations that are based on lifting~\cite{bpp,bppip}. However, these separations seem to require quantum protocols with at least two rounds of communication, so they imply a separation of two-way quantum and classical communication complexity but they do not give the stronger separations of simultaneous-message quantum communication complexity vs. two-way classical communication complexity (or even one-way quantum communication complexity vs. two-way classical communication complexity).

Our proof technique is completely new, in the context of communication complexity,  and is based on techniques from~\cite{raztal}. Our function $f$ is based on a lift of the \textsc{forrelation} problem, using \textsc{xor} as a gadget.

\end{abstract}

\section{Introduction}

Exponential separations between quantum and classical communication complexity have been established in various models and settings. These separations give explicit examples of partial functions that can be computed by quantum protocols with very small communication complexity, while any classical randomized protocol requires significantly higher communication complexity. However, to the best of our knowledge, in all these works the efficiency of the quantum players in the quantum protocol has not been addressed and in most of  these separations, the quantum players are inefficient.

Communication complexity studies the amount of communication needed to
perform computational tasks that depend on two (or more) inputs, each given to a different player. The efficiency of the players in a communication complexity protocol is usually not addressed. If
the players need to read their entire inputs, their time complexity is at least the length of the inputs.
However,
the inputs may be represented compactly by a black box and (particularly in the quantum case) we can hope for
players that can be implemented very efficiently by small (say, poly-logarithmic size) quantum circuits, with oracle access to their inputs.

Our main result qualitatively matches the strongest known separation between quantum and classical communication complexity~\cite{gavinsky} and is obtained using quantum protocols where all players are efficient. To prove our results we use a completely different set of techniques, based on techniques from the recent oracle separation of BQP and PH~\cite{raztal}.

\subsection{Previous Work}

The relative power of quantum and classical communication complexity has been studied in numerous of works. While it is unknown whether quantum communication can offer exponential advantage over randomized communication for total functions, a series of works
gave explicit examples of partial Boolean functions (promise problems)
that have quantum protocols with very small communication complexity,
while any classical protocol requires exponentially higher communication complexity. The history of exponential advantage of quantum communication, that is most relevant to our work, is briefly summarized below.

Buhrman, Cleve and Wigderson gave the first (exponential) separation between zero-error quantum communication complexity and classical deterministic communication complexity~\cite{buhrman}. Raz gave the first exponential separation between two-way quantum communication complexity and two-way randomized communication complexity~\cite{raz}. Bar-Yossef et al~\cite{Bar-YossefJK04} (for search problems) and
Gavinsky et al~\cite{gavinskyetal} (for promise problems) gave the first (exponential) separations between one-way quantum communication complexity and one-way randomized communication complexity. Klartag and Regev gave the first (exponential) separation between one-way quantum communication complexity and two-way randomized communication complexity~\cite{klartagregev}. Finally, Gavinsky gave an (exponential) separation between simultaneous-message quantum communication complexity 
and two-way randomized communication complexity~\cite{gavinsky}. 

We note that Gavinsky's work is the strongest separation known today and essentially subsumes the separations discussed above.
More precisely, Gavinsky~\cite{gavinsky} gave an explicit partial Boolean function~$f$ over inputs of length $N$, such that:
\begin{enumerate}
\item
$f$ can be computed by a simultaneous-message quantum protocol with communication complexity $\mbox{polylog}(N)$: 
Alice and Bob simultaneously send quantum messages of length $\mbox{polylog}(N)$ to a referee, who performs a quantum measurement on the messages and announces the answer. (At the beginning of the protocol Alice and Bob also have  $\mbox{polylog}(N)$ entangled EPR pairs).

We note that this also implies a one-way quantum protocol where Alice sends a message of length $\mbox{polylog}(N)$ qubits to Bob, who performs a measurement and announces the answer (or vice versa).
\item
Any classical randomized protocol for $f$ has communication complexity at least $\Omega\left(N^{1/2}\right)$.
\end{enumerate}

A drawback of Gavinsky's separation, in the context of our work, is that the referee in his quantum protocol is inefficient as it is required to perform $O(N)$ quantum operations (and this seems to be crucial in his lower bound proof).

As mentioned before, to the best of our knowledge, the efficiency of the quantum players has not been addressed in previous works on separations of quantum and classical communication complexity.
The only separations
that we know of that do have efficient quantum parties are the separations that follow from the recent randomized query-to-communication lifting theorems of~\cite{bpp,bppip}, applied to problems for which we know that quantum decision trees offer an  exponential advantage over randomized ones, 
such as the \textsc{forrelation} problem of~\cite{aaronson10,aaronsonambainis}.
However, lifting with the gadgets used in~\cite{bpp,bppip} seems to require quantum protocols with two rounds of communication. Thus, these theorems only imply a separation of two-way quantum and classical communication complexity and do not give the stronger separations of simultaneous-message quantum communication complexity vs. two-way classical communication complexity (or even one-way quantum communication complexity vs. two-way classical communication complexity).

\subsection{Our Result}

We recover Gavinsky's state of the art separation, using entirely different techniques. While the parameters in our bounds are weaker, our quantum protocol is {\it efficient}, in the sense that it involves just $\mbox{polylog}(N)$ amount of work by Alice, Bob and the referee,  when the players have blackbox access to their inputs. In other words, the output of the entire simultaneous protocol can be described by a $\mbox{polylog}(N)$ size quantum circuit,
with oracle access to the inputs.

More precisely, our main result gives an explicit partial Boolean function~$f$ over inputs of length $N$, such that:
\begin{enumerate}
\item
As in Gavinsky's work,
$f$ can be computed by a simultaneous-message quantum protocol with communication complexity $\mbox{polylog}(N)$:
Alice and Bob simultaneously send quantum messages of length $\mbox{polylog}(N)$ to a referee, who performs a quantum measurement on the messages and announces the answer. (At the beginning of the protocol Alice and Bob also have  $\mbox{polylog}(N)$ entangled EPR pairs).

As before, this also implies a one-way quantum protocol where Alice sends a message of length $\mbox{polylog}(N)$ qubits to Bob, who performs a measurement and announces the answer (or vice versa).
\item
Any classical randomized protocol for $f$ has communication complexity at least $\tilde{\Omega}\left(N^{1/4}\right)$.
\item
All parties in the quantum protocol of Item~(1) (Alice, Bob and the referee) can be implemented by quantum circuits of size $\mbox{polylog}(N)$ (where Alice and Bob have oracle access to their input).
\end{enumerate}

The problem that we define is a lift of the \textsc{forrelation} problem of~\cite{aaronson10,aaronsonambainis,raztal} with \textsc{xor} as the gadget.
Our proof technique follows the Fourier-analysis framework of~\cite{raztal}. Our proof offers an entirely new and possibly simpler approach for communication complexity lower bounds. We believe this technique may be applicable in a broader setting.
We note that lower bounds for lifting by \textsc{xor}, using a Fourier-analysis approach, were previously studied in~\cite{DBLP:journals/cc/Raz95,DBLP:journals/siamcomp/HatamiHL18}.

\subsection{Our Communication Complexity Problem}

Let $N=2^n$ and $H_N$ be the $N\times N$ normalized Hadamard matrix. Let $x=(x_1,x_2)$ be an input where $x_1,x_2\in \{-1, 1\}^N$. The forrelation $forr(x)$ of a vector $x$ is defined as follows and measures how correlated the second half is with the Hadamard transform of the first half.
\[forr(x):=\frac{1}{N} \langle H_N(x_1)|x_2\rangle \]
The communication problem for which our separation holds is a lift of the forrelation problem of~\cite{raztal}, with XOR as the gadget. Let $x,y \in \{-1, 1\}^{2N}$. Alice gets $x$ and Bob gets $y$ and their goal is to compute the partial function $F$ defined by
\[F(x,y):=\begin{cases} 1 & \text{ if } forr(x\cdot y ) \ge  \frac{1}{200} \cdot \frac{1}{\ln N} \\ -1 & \text{ if } forr(x\cdot y) \le \frac{1}{400} \cdot \frac{1}{\ln N} \end{cases} \]
Here $x\cdot y$ refers to the coordinate-wise product of the vectors $x,y$. The quantum upper bound on $F$ follows from the fact that the XOR of the inputs can be computed by a simultaneous-message quantum protocol, when the players share entanglement, and the fact that $forr(x)$ can be estimated by a small size quantum circuit~\cite{aaronson10,aaronsonambainis,raztal}.

\subsection{An Overview of the Lower Bound}
We briefly outline the proof of the lower bound. We use the forrelation distribution $\mathcal{D}$ on $\{-1,1\}^{2N}$ as defined by ~\cite{raztal}. We define a distribution $\mathcal{V}$ on inputs to the communication problem, obtained by sampling $z\sim\mathcal{D}$, and $x\in \{-1,1\}^{2N}$ uniformly at random, and setting $y:=x\cdot z$. Alice gets $x$ and Bob gets $y$. It can be shown that the distribution $\mathcal{V}$ has considerable support over the yes instances of $F$, while the uniform distribution $\mathcal{U}$ on $\{-1,1\}^{4N}$ has large support over the no instances of $F$. This fact along with the following theorem implies a lower bound on the randomized communication cost of $F$.

{\bf Theorem [Informal]:}  {\it No deterministic protocol of cost $o(N^{1/4})$ has considerable advantage in distinguishing $\mathcal{V}$ from $\mathcal{U}$. } 

We now outline the proof of this theorem. Any cost $c$ protocol induces a partition of the input space into at most $2^c$ rectangles. Let $A\times B$ be any rectangle, and let $\mathbbm{1}_A, \mathbbm{1}_B: \{-1,1\}^{2N}\rightarrow \{0,1\}$ be the indicator functions of $A$ and $B$ respectively. Note that for all distributions $\mathcal{S}$ on $\{-1,1\}^{2N}$, we have
\[ \underset{z\sim \mathcal{S},}{\E} \underset{x\sim U_{2N}}{\E} \left[\mathbbm{1}_A(x)\mathbbm{1}_B(x\cdot z)\right]=\underset{z\sim \mathcal{S}}{\E}  \left[(\mathbbm{1}_A * \mathbbm{1}_B)(z)\right] \]

Here, the notation $f*g$ refers to the convolution of Boolean functions $f$ and $g$. This identity implies that our goal is to show that the expectation of the function $(\mathbbm{1}_A * \mathbbm{1}_B)(z)$ over a uniformly distributed $z$ is close to the expectation over $z\sim\mathcal{D}$. An essential contribution of the works of \cite{raztal} and~\cite{chlt} is the following result.  For any family of functions $\mathcal{F}$ that is closed under restrictions, to show that the family is fooled by the forrelation distribution, it suffices to bound the $\ell_1$-norm of the second level Fourier coefficients of the family. More precisely, the maximum advantage of a function $f\in \mathcal{F}$ in distinguishing the uniform distribution and $\mathcal{D}$, is at most $O\left(\frac{1}{\sqrt{N}}\right)$ times the maximum second level Fourier mass of a function $f\in \mathcal{F}$. Since small cost communication protocols form a family of functions closed under restrictions, the same reasoning applies here.

We now describe how to bound the second level Fourier mass corresponding to a small cost protocol. Let $A\times B$ be a rectangle. An important property of the convolution of two functions $f,g$ is that for all subsets $S\subseteq [n]$, we have $\widehat{f*g}(S)=\widehat{f}(S)\widehat{g}(S)$. This, along with Cauchy-Schwarz implies that
\[  \sum_{|S|=2} \left| \widehat{\mathbbm{1}_A * \mathbbm{1}_B}(S) \right| = \sum_{|S|=2} \left| \widehat{\mathbbm{1}_A}(S) \widehat{\mathbbm{1}_B}(S)\right| \le  \left( \sum_{|S|=2}\widehat{\mathbbm{1}_A}(S)^2 \right)^{1/2} \left( \sum_{|S|=2}\widehat{\mathbbm{1}_B}(S)^2 \right)^{1/2} \]
We then use a well known inequality on Fourier coefficients. It appears as `Level-k Inequalities' in Ryan Odonnell's book \cite[Chapter~9.5]{odonnell} and it states that for a function $f:\{-1,1\}^n\rightarrow \{0,1\}$ with expectation $\E[f]=\alpha$, for any $k\le 2\ln (1/\alpha)$, we have $ \sum_{|S|=k} \left( \widehat{f}(S)\right)^2 \le O(\alpha^2 \ln^k(1/\alpha))$. For simplicity, assume that $|A|=|B|=2^{(n-c)/2}$. The previous paragraphs and the assumption that $\E[\mathbbm{1}_A], \E[\mathbbm{1}_B]=  \frac{1}{2^{c/2}}$ imply that the advantage of a rectangle is at most $O\left(\frac{1}{\sqrt{N}}\frac{1}{2^c}c^2\right).$ Adding the contributions from all rectangles implies that the advantage of a cost $c$ protocol is at most $O\left(\frac{c^2}{\sqrt{N}}\right)$. This implies that every protocol of cost $o(N^{1/4})$ has advantage at most $o(1)$ in distinguishing between $\mathcal{U}$ and $\mathcal{V}$. The bound in the case of a general partition follows from a concavity argument. This completes the proof overview.

We conjecture that the correct randomized communication complexity for this problem is $\tilde{\Omega}(\sqrt{N})$ and that the above proof technique can be strengthened to show this. One way to do this would be to show a better bound on the Fourier coefficients of deterministic communication protocols. In particular, it would suffice to show a bound of $O(c\cdot  \mathrm{poly}\log(N))$ on the second level Fourier mass of protocols with $c$-bits of communication.

\section{Preliminaries} 
For $n\in \mathbb{N}$, let $[n]$ denote the set $\{1,2,\ldots,n\}$. For a vector $x\in \mathbb{R}^n$ and $i\in [n]$, we refer to the $i$-th coordinate of $x$ by either $x(i)$ or $x_i$. For a subset $S\subset [n]$, let $x_S\in\mathbb{R}^{|S|}$ be the restriction of $x$ to coordinates in $S$. For vectors $x,y\in \mathbb{R}^n$, let $x\cdot y$ be their point-wise product, i.e., the vector whose $i$-th coordinate is $x_i y_i$. Let $\langle x |y\rangle$ be the real inner product $\sum_i x_i y_i$ between $x$ and $y$. Let $v^{-1}$ be the coordinate-wise inverse of a vector $v\in (\mathbb{R}\setminus 0)^n$. 

\subsection{Fourier Analysis on the Boolean Hypercube}
The set $\{-1, 1\}^n$ is referred to as the Boolean hypercube in $n$ dimensions, or the $n$-dimensional hypercube. We  sometimes refer to it by $\{0,1\}^n$, using the bijection mapping $(x_1,\ldots,x_n)\in \{0,1\}^n$ to $((-1)^{x_1},\ldots,(-1)^{x_n})\in \{-1, 1\}^n$. We also represent elements of $\{-1, 1\}^n$ by elements of $[2^n]$, using the bijection mapping $((-1)^{x_1},\ldots,(-1)^{x_n})\in \{-1,1\}^n$ to $1+\sum_{i=1}^n 2^{i-1} x_i \in [2^n]$. We typically use $N$ to denote $2^n$. Let $\mathbb{I}_n$ denote the $n\times n$ identity matrix. Let $U_{n}$ be the uniform distribution on $\{-1, 1\}^n$. Let $\mathcal{F}:=\{F: \{-1, 1\}^n\rightarrow \mathbb{R}\}$ be the set of all functions from the $n$-dimensional hypercube to the real numbers. This is a real vector space of dimension $2^n$. We define an inner product over this space. For every $f,g, \in \mathcal{F}$, let
\[ \langle f, g \rangle : =  \underset{x\sim U_n}{\mathbb{E}} \left[f(x)g(x)\right]\] 
For any universe $\mathcal{U}$ and a subset $S\subseteq \mathcal{U}$, we use $\mathbbm{1}_S : \mathcal{U}\rightarrow \{0,1\}$ to refer to the indicator function of $S$ defined by:
\[  \mathbbm{1}_S(x):= \begin{cases} 1 &\text{ if } x\in S \\ 0 &\text{ otherwise } \end{cases} \]
The set of indicator functions of singleton sets $\{\mathbbm{1}_{\{a\}}: a\in \{-1, 1\}^n\}$ is the standard orthogonal basis for $\mathcal{F}$. The character functions form an orthonormal basis for $\mathcal{F}$. These are functions $\chi_S:\{-1,1\}^n\rightarrow\{-1,1\}$ associated to every set $S\subseteq [n]$ and are defined at every point $x\in \{-1,1\}^n$ by
$\chi_S(x):= \prod_{i\in S} x_i$. For a function $f\in \mathcal{F}$, and $S\subseteq[n]$, we define its $S$-th Fourier coefficient to be $\widehat{f}(S):=\E_{x\sim U_n} \left[f(x)\chi_S(x)\right]$. Every $f\in \mathcal{F}$ can be expressed as $f(x)=\sum_{S\subseteq [n]} \widehat{f}(S)\chi_S(x)$.  For $f:\{-1, 1\}^n\rightarrow \mathbb{R}$ and $k\in\{0,\ldots,n\}$, let $L_{k}(f):=\sum_{S\subseteq [n],|S|=k}\left|\widehat{f}(S)\right|$ refer to the level $k$ Fourier mass of $f$.

Given functions $f,g:\{-1, 1\}^n\rightarrow \mathbb{R}$, their convolution $f*g:\{-1,1\}^n\rightarrow \mathbb{R}$ is defined as $ f *g(x):= \underset{y\sim U_n}{\E}\left[ f(y) g(y\cdot x) \right] $.  A standard fact about convolution of functions is that $\widehat{f*g}(S)=\widehat{f}(S)\widehat{g}(S)$ for all $S\subseteq [n]$.

\subsection{Quantum Computation}

Let $\mathcal{H}_m$ be the Hilbert space of dimension $2^m$ defined by the complex span of the orthonormal basis $\{\ket{x}:x\in\{-1,1\}^m\}$. We sometimes express these basis elements by integers $\{\ket{i} :i\in[2^m]\}$ by the same correspondence as before. An element in this space is denoted by $\ket{\phi}$ and is a unique complex combination of the vectors $\ket{x}$, where $x$ is a bit string in $\{-1,1\}^m$. We omit the subscript on $\mathcal{H}$ when it is implicit. Pure quantum states on $m$ qubits are described by unit vectors in $\mathcal{H}_m$. We sometimes use the terms register and qubit interchangeably. Note that we have the vector space isomorphism $\mathcal{H}_m\cong \otimes_{i=1}^m \mathcal{H}_1^{(i)}$, where each $\mathcal{H}_1^{(i)}\cong \mathcal{H}_1$. We call $\mathcal{H}_1^{(i)}$ the $i$-th register, or the $i$-th qubit. The evolution of a pure state can be described by either projective or unitary transformations on $\mathcal{H}_m$, a few of which we describe below. 
\begin{itemize}
\item \textsc{Unitary operators}: Unitary operators over $\mathcal{H}_m$ act on pure states in the natural way. They map a pure state $\ket{\phi}$ to a pure state $U(\ket{\phi})$.
\item \textsc{Hadamard operator} $H$: The Hadamard matrix $H_N$ is an $N\times N$ unitary matrix acting on $\mathcal{H}_n$. We let $x$ and $y$ in $\{0,1\}^n$ index rows and columns of $H_N$ respectively. The entries of $H_N$ are as follows.
\[ H_N(x,y) := \begin{cases} \frac{1}{\sqrt{N}}  & \text{ if }  \sum_i x_iy_i\mod 2=0 \\
\frac{-1}{\sqrt{N}} &\text{otherwise} \end{cases}\]
We refer to the single bit unitary operator $H_1$ as simply $H$. We have the identity $H_{2^n}=H^{\otimes n}$. Thus, the action of $H_{2^n}$ on $\mathcal{H}_n$ can be described as the tensor product of the actions of $H$ on $\mathcal{H}_1^{(i)}$ for $i\in [n]$.
\item \textsc{Controlled not}: This is a two-bit unitary operator, where the first register is the {\it control} and the second is the {\it target}. For $x_1,x_2 \in \{\pm 1\}$, it maps $\ket{(x_1,x_2)}$ to $\ket{(x_1,-x_2)}$  if $x_1=-1$ and otherwise leaves it fixed.
\item \textsc{Clifford} operator $R_{\pi/8}$: This is a single qubit unitary operator given by the 2 $\times$ 2 unitary matrix $\begin{bmatrix} \cos \frac{\pi}{8} & -\sin\frac{\pi}{8} \\ \sin\frac{\pi}{8} & \cos\frac{\pi}{8} \end{bmatrix}$.
\item \textsc{Measurement} of the $i$-th Register: Let $M_{i,1}$ (respectively $M_{i,-1}$) be the projection operator onto the span of $\{ \ket{x} : x(i)=1 \}$ (respectively the span of $\{ \ket{x} : x(i)=-1 \}$). The measurement of the $i$-th register of a pure state $\ket{\phi}$ is a probabilistic process which returns the state $\frac{ M_{i,b}\ket{\phi} }{ \| M_{i,b}\ket{\phi}\| }$ with probability $\| M_{i,b} \ket{\phi}\|^2$ for $b\in \{-1,1\}$. The subsequent value of $b$ is said to be the {\it outcome} of the measurement.
\end{itemize}

\begin{figure}
\centering
\mbox{ 
\Qcircuit @C=1em @R=.7em {
& \gate{H}  & \qw &  &  \meter & \qw & &\targ & \qw\\
}}
\caption{Representation of the \textsc{Hadamard}, the \textsc{measurement} and \textsc{not} operators. A horizontal wire represents a register and a labeled box an operator.}
\end{figure}
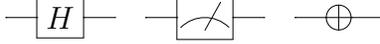

\begin{figure}
\centering
\mbox{ 
\Qcircuit @C=1em @R=.7em {
& \ctrl{1}  & \qw   \\
& \targ & \qw \\
}}
\caption{Representation of the $CNOT$ operator. The black dot represent the control register and $\oplus$ represents the target register.}
\end{figure}
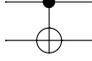

A quantum circuit $Q:\{-1, 1\}^n\rightarrow \{-1, 1\}^m$ of space $S$ consists of a set of $S$ registers, the first $n$ of which are initialized to $\ket{x}$, the input, while the rest are initialized to $\ket{1}$. It further consists of a sequence of operators chosen from $\{H, CNOT_{i,j}, R_{\pi/8}, M_i \}$ along with a description of which register they act on. The size of a circuit is the number of operators. The output of a circuit is defined to be the contents of the first $m$ registers. Since we want the output to be Boolean, we assume that the circuit measures these registers and returns the outcome. Thus, a quantum circuit is inherently probabilistic.

We now describe quantum circuits with {\it query or oracle} access. In this model, all registers are initialized to $\ket{1}$ and the input $x\in \{-1,1 \}^n$ is not written into the registers. Instead, it is compactly presented to the algorithm using a blackbox, a device which for every index $i\in [n]$, returns $x(i)\ket{i}$ when it is given $\ket{i}$ as input. More precisely, for every possible input $x\in \{-1, 1\}^n$, the oracle to $x$ is the linear operator $O_x: \mathcal{H}_{\lceil \log n\rceil} \rightarrow \mathcal{H}_{\lceil \log n\rceil}$ which maps the basis states $\ket{i}$ to $x_i\ket{i}$ whenever $i\in [n]$ and otherwise leaves it fixed. This indeed restricts to a unitary operation on pure states, as its action on the basis states is described by a diagonal $\{-1, 1\}$-matrix. This serves as the quantum analogue of a classical oracle, which is a blackbox that returns $x(i)$ on input $i\in [n].$ A {\it quantum circuit with oracle access to inputs} is a quantum circuit that is allowed to use the $O_x$ operator in addition to the usual operators, where $x$ is the input to the computation. The {\it size} of the circuit is the total number of operators from $\{H, CNOT_{i,j}, R_{\pi/8}, M_i , O_x \}$ used. We say that an algorithm is {\it efficient}, if it is described by a circuit of size at most $poly\log n$ with oracle access to inputs.

Note that it is possible to use the oracle $O_x$ to explicitly write down the input $x$ into $n$ registers, however, this requires $n$ oracle calls and $n$ registers. It is often the case that this step is unnecessary. 

\subsection{Classical \& Quantum Communication Complexity}

Let $f:\{-1,1\}^n \times \{-1,1\}^m \rightarrow \{-1,1\}$ be a partial Boolean function. Alice (respectively Bob) receives a private input $x\in \{-1,1\}^n$ (respectively $y\in \{-1,1\}^m$) and the players' goal is to compute $f(x,y)$ if $(x,y)$ is in the support of $f$, while exchanging as few bits as possible. An input $(x,y)$ is said to be a \textsc{yes} (respectively \textsc{no}) instance if $f(x,y)=-1$ (respectively if $f(x,y)=1$).

A {\it deterministic communication protocol} $D$ proceeds in rounds, and in each round, a player sends the other a message in $\{-1, 1\}$. A message sent by a player in a given round is the output bit of some fixed Boolean function of their private input and the messages they received in the previous rounds. At the end, Alice returns a bit $D(x,y)$, the output of the protocol. The protocol computes $f$ if for all $(x,y)$ in the support of $f$, we have $D(x,y)=f(x,y)$. The {\it communication cost} of the protocol is the maximum over the inputs $(x,y)$ in the support of $f$ of the number of bits exchanged. We assume that the protocol returns a bit in $\{-1,1\}$ even when run on inputs not in the support of $f$, this can be done by aborting and returning $1$ if the players realize that their inputs are not in the support of $f$. The sequence of messages is called the {\it transcript}. For every protocol of cost at most $c$ and for every possible transcript in $\{-1,1\}^c$, the set of input pairs $(x,y)\in\{-1,1\}^n\times \{-1,1\}^m$ that could have generated this transcript is a rectangle $A\times B$, where $A$ (respectively $B$) is the set of Alice's (respectively Bob's) inputs that could have generated the transcript. Thus, every deterministic protocol of cost at most $c$ induces a partition of the input space $\{-1,1\}^n\times \{-1,1\}^m$ into at most $2^c$ rectangles. 

In a {\it bounded-error randomized protocol} $C$, Alice and Bob have access to a shared unbiased coin which they can toss arbitrarily many times. Based on the outcome $r$ of the coin tosses, they run a deterministic protocol $D_r$. The protocol is said to compute $f$ with error at most $\epsilon$ if for each $(x,y)$ in the support of $f$, with at least $1-\epsilon$ probability over the coin tosses, the output of the deterministic protocol equals $f(x,y)$. The cost is the maximum cost of the deterministic protocols. The min-max principle states that for any partial function $f:\{-1,1\}^n\times \{-1,1\}^m\rightarrow \{-1,1\}$, there is a bounded error protocol of cost at most $c$ computing $f$ with error at most $\epsilon$ if and only if for all distributions $\mu$ on the support of $f$, there is a deterministic protocol $D$ of cost at most $c$ such that $\underset{(x,y)\sim \mu}{\mathbb{E}} D(x,y)f(x,y) \ge 1-2\epsilon$. We assume $\epsilon=1/3$ by default.

In the quantum communication model, Alice and Bob have infinitely many private qubits, the first of which are initialized to their respective inputs and the rest to $\ket{1}$. A {\it quantum bounded-error protocol} $Q$ consists of several rounds and in each round, a player applies a unitary or a measurement operator to qubits they own and then sends a qubit to the other player. The sequence of operators and the qubits they act on is fixed beforehand. At the end of the protocol, Alice returns a bit $Q(x,y)$, the output of the protocol. The protocol is said to compute $f$ if for every $(x,y)$ in the support of $f$, with probability at least $2/3$, the output $Q(x,y)$ equals $f(x,y)$. The cost of the protocol is the maximum number of qubits exchanged.

In communication with {\it entanglement}, the players are given the additional resource of entangled qubits. These are $2m$ registers for some $m\in \mathbb{N}$, the first half of which belong to Alice and the second half to Bob. The registers are initialized to the state $\frac{1}{\sqrt{2^m}} \sum_{i=1}^{2^m} \ket{i}^A\ket{i}^B$. Here, the superscript indicates to which player the register belongs. The assumption on the initial entangled state is natural as this state is obtained by tensoring $m$ independent copies of the Bell state $\frac{1}{\sqrt{2}}\big(\ket{0}^A\ket{0}^B + \ket{1}^A\ket{1}^B\big)$. In other words, it is as if Alice and Bob had $m$ independent copies of the Bell state. We say that in a protocol using $\frac{1}{\sqrt{2^m}} \sum_{i=1}^{2^m} \ket{i}^A\ket{i}^B$ as the initial state, Alice and Bob share $m$ bits of entanglement.

In the {\it simultaneous} model of communication, Alice and Bob are not allowed to exchange registers with each other. Instead, they are allowed one round of communication with a referee Charlie, to whom they can only send qubits. The referee then performs some quantum operation on the qubits he receives and returns a bit as the output. As before, a bounded-error simultaneous protocol computes $f$ if for all $(x,y)$ in the support of $f$, with probability at least 2/3, the referee's output agrees with $f(x,y)$. The cost is the total number of qubits that Alice and Bob send the referee. 

Note that in each of the above models of communication, every function $f:\{-1,1\}^n\times \{-1,1\}^m  \rightarrow \{-1,1\}$ has communication cost at most $n+m$, since the players may simply reveal their entire inputs. Hence, a small cost protocol is one in which the communication cost is at most $poly\log(n+m)$.

A communication protocol is said to be {\it efficient} if it can be implemented by a small size circuit with oracle access $O_x,O_y$ to the inputs $x,y$. Protocols with small communication cost are not necessarily efficient, as they may require computationally intensive processing on the messages, or they may require the players to make several probes into their inputs. 

\subsection{The Forrelation Distribution $\mathcal{D}$}

Let $x\sim \mathcal{D}$ refer to a random variable $x$ distributed according to the probability distribution $\mathcal{D}$. We use $\p_\mathcal{D}$ to refer to the probability measure associated with $\mathcal{D}$ and $\p_{x\sim \mathcal{D}}(E(x))$ to refer to the probability of event $E(x)$ when $x\sim \mathcal{D}$. For an event $E(x)$, we will denote by $\mathcal{D}|E(x)$ (respectively $\mathcal{D}|\neg E(x)$), the distribution $\mathcal{D}$ conditioned on the event $E(x)$ occurring (respectively, the event $E(x)$ not occurring). Let $\epsilon\ge0$ be a parameter, $f(x):\mathbb{R}^{n}\rightarrow \mathbb{R}$ a function and $\mathcal{D}$ a distribution on $\mathbb{R}^{n}$. We say that $\mathcal{D}$ fools $f$ with error $\epsilon$ if $\left|\underset{x\sim U_{n}}{\E}[f(x)]-\underset{x\sim \mathcal{D}}{\E}[f(x)]\right|\le \epsilon$.

Let $\mathcal{N}(\mu,\sigma^2)$ denote a Gaussian distribution of mean $\mu \in \mathbb{R}$ and variance $\sigma^2\in \mathbb{R}_{\ge 0}$. We will repeatedly use the following standard facts about Gaussians.
\begin{itemize}
\item Gaussian Concentration inequality: For $X\sim \mathcal{N}(\mu,\sigma^2)$, we have $ \p[|X-\mu|\ge a]\le e^{-\frac{a^2}{2\sigma^2}}$.
\item The sum $\sum_i X_i$ of independent Gaussians $X_i \sim \mathcal{N}(\mu_i,\sigma_i^2)$ is distributed according to  $ \mathcal{N}(\sum_i \mu_i, \sum_i \sigma_i^2)$.
\end{itemize} 

We will also use Chebyshev's inequality and Chernoff's bound. Chebyshev's inequality~\cite{chebyshev} implies that for a set of $n$ pair-wise independent random variables $X_i$ with mean $\mu_i$ and variance $\sigma_i^2$ , we have $\p[ |\sum_{i=1}^n (X_i-\mu_i) | \ge a] \le \frac{\sum_{i=1}^n\sigma_i^2}{a^2}$. Chernoff's bound~\cite{bernstein,cos521} implies that for $n$ independent identical random variables $X_i$ in $[-1,1]$ whose sum is of mean $\mu$ and variance $\sigma^2$, we have $\p \left[\left|  \sum_{i=1}^n X_i - \mu \right| \ge t\sigma \right]\le 2\exp(-t^2/4)$ whenever $t\le \frac{\sigma}{2}.$ 

Let $x=(x_1,x_2)$ for $x_1,x_2 \in \{-1, 1\}^{N}$. We define the forrelation of $x$ as the correlation between the second half $x_2$ and the Hadamard transform of the first half $x_1$. 
\[ forr(x) := \left\langle   \frac{1}{\sqrt{N}} H_N(x_1)\Bigg| \frac{1}{\sqrt{N}} x_2\right\rangle \] 

We state the definition of the forrelation distribution, as defined in~\cite{raztal}. Fix a parameter $\epsilon=\frac{1}{50 \ln N}$. We first define an auxilliary Gaussian distribution $\mathcal{G}$ generated by sampling the first half uniformly at random and letting the second half be the Hadamard transform of the first half. More precisely,
\begin{enumerate}
\item Sample $x_1,\ldots,x_N\sim \mathcal{N}(0,\epsilon)$.
\item Let $y=H_N x$.
\item Output $(x,y)$. 
\end{enumerate}
This is a Gaussian random variable in $2N$ dimensions of mean $0$ and covariance matrix given by
\[\epsilon \begin{bmatrix} \mathbb{I}_N & H_N \\ H_N & \mathbb{I}_N \end{bmatrix} \]

Let $trnc:\mathbb{R} \rightarrow [-1,1]$ be the truncation function which on input $\alpha>1$, returns 1, $\alpha<-1$ returns $-1$ and otherwise returns $\alpha$. This naturally defines a function $trnc:\mathbb{R}^{2N}\rightarrow [-1,1]^{2N}$ obtained by truncating each coordinate. We now define a distribution $\mathcal{D}$ over $\{-1, 1\}^{2N}$ generated from $ \mathcal{G}$ by truncating the sample and then independently sampling each coordinate as follows.
\begin{enumerate}
\item Sample $z\in \mathcal{G}$.
\item For each coordinate $i\in [2N]$ independently, let $z'_i=1$ with probability $\frac{1+trnc(z_i)}{2}$ and $-1$ with probability $\frac{1-trnc(z_i)}{2}$. 
\item Output $z'$.
\end{enumerate}
\sloppy We refer to the distribution $\mathcal{D}$ as the {\it forrelation} distribution. We state Claim 6.3 from~\cite{raztal} which implies that a vector drawn from this distribution has large forrelation on expectation. The proof is omitted.
\begin{lemma}\label{raztallemma} Let $\mathcal{D}$ be the forrelation distribution as defined previously. Then,
\[\E_{z\sim \mathcal{D}} [forr(z)]\ge \frac{\epsilon}{2} \]
\end{lemma}

\subsection{Multilinear Functions on $\mathcal{D}$}

Given a function $f:\{-1,1\}^n\rightarrow \mathbb{R}$, there is a unique multilinear polynomial $\tilde{f}:\mathbb{R}^n\rightarrow \mathbb{R}$ which agrees with $f$ on $\{-1, 1\}^n$. This polynomial is called the multilinear extension of $f$. The multilinear extension of any character function $\chi_S(x)$ is precisely $\prod_{i\in S} x_i$. The multilinear extension $\tilde{f}$ of $f$ satisfies $\tilde{f}(x)=\sum_{S\subseteq [n]} \widehat{f}(S) \prod_{i\in S}x_i$ for all $x\in \mathbb{R}^n$. We sometimes identify $f$ with its multilinear extension. The main content of this section is that bounded multilinear functions have similar expectations under $\mathcal{G}$ and under $\mathcal{D}$.  

\begin{claim} \label{claim1}   Let $F:\mathbb{R}^{2N}\rightarrow \mathbb{R}$ be any multinear function $F=\sum_S \widehat{F}(S)\chi_S$. Then,
\[ \underset{z'\sim \mathcal{D}}{\E}[F(z')]= \underset{z\sim \mathcal{G}}{\E}[F(trnc(z))]  \]
\end{claim}

\begin{proof}[Proof of Claim \ref{claim1}] For any fixed $z$, recall that the sampling process for $\mathcal{D}$ involves independently setting $z_i'$ to $1$ with probability $\frac{1+trnc(z)_i}{2}$ and to $-1$ with the remaining probability. Because of this and linearity of expectation, we have
\begin{align*}
\begin{split}
\E\left[F(z') \mid z\right]&=\E\Big[\sum_S \widehat{F}(S)\chi_S(z') \Big|  z\Big]=\sum_S \widehat{F}(S)\E[\chi_S(z')\mid z]\\
&=\sum_S \widehat{F}(S) \chi_S(\E[z'|z]) \\
&=\sum_S \widehat{F}(S)\chi_S(trnc(z))=F(trnc(z))\end{split} \end{align*}
\end{proof}

This implies that $\E_{z'\sim \mathcal{D}}[F(z')]$ is exactly $\E_{z\sim \mathcal{G}}[F(trnc(z))]$. The following claim states that $\E_{z\sim \mathcal{G}}[F(trnc(z))]$ is pretty close to $\E_{z\sim \mathcal{G}}[F(z)]$ for a bounded multilinear function $F$. Its proof is identical to that in \cite{raztal}, so we omit it. The underlying idea is that $\epsilon$ is small, so the random variable $z\sim \mathcal{G}$ has an exponentially decaying norm, furthermore, bounded multilinear functions $F$ on $\{-1, 1\}^{2N}$ cannot grow faster than exponentially in the norm of the argument.
\begin{claim}\label{claim3} Let $F(z)$ be any multilinear polynomial mapping $\{-1, 1\}^{2N}$ to $[-1,1]$. Let $z_0\in [-1/2,1/2]^{2N}$, $p\le \frac{1}{2}$ and $N>1$. Then,
\[ \E_{z\sim \mathcal{G}} \left[\left| F(trnc(z_0+pz))-F(z_0+pz) \right|\right] \le \frac{8}{N^5} \]
\end{claim}
We remark that the bound in \cite{raztal} is $\frac{8}{N^2}$. The improved bound of $\frac{8}{N^5}$ in Claim \ref{claim3} follows from our choice of $\epsilon=\frac{1}{50\ln N}$, as opposed to $\epsilon=\frac{1}{24\ln N}$ as in \cite{raztal}.

\subsection{Moments of $\mathcal{G}$}
In this section we state some facts about the moments of the forrelation distribution that will be useful later. We use the following notation to refer to the moments of $\mathcal{G}$.
\[\widehat{\mathcal{G}}(S,T):= \underset{(x,y)\sim \mathcal{G}}{\E} \left[ \prod_{i\in S}x_i \prod_{j\in T}y_j \right] \]
The following claim and its proof are analogous to Claim 4.1 in ~\cite{raztal}.
\begin{claim} \label{claim4} Let $S,T\subseteq [N]$ and $i,j\in [N]$. Let $k_1=|S|,k_2=|T|$. Then,
\begin{enumerate}
\item $\widehat{\mathcal{G}}(\{i\},\{j\})=\epsilon N^{-1/2}(-1)^{\left<i,j\right>}$.
\item $\widehat{\mathcal{G}}(S,T)=0$ if $k_1\neq k_2$.
\item $\left|\widehat{\mathcal{G}}(S,T)\right|\le \epsilon ^{k} k!N^{-k/2}$ if $k=k_1=k_2$.
\item $\left|\widehat{\mathcal{G}}(S,T)\right|\le \epsilon^{|S|}$ for all $S,T$.
\end{enumerate} 
\end{claim}

\section{The Forrelation Communication Problem}

In this section we formally state the main theorems of this paper. Their proofs follow in the successive sections.

Let $\epsilon=\frac{1}{50 \ln N}$ be the parameter as before, defining the forrelation distribution.

\begin{theorem} \label{theorem1}  Consider the following distribution. A string $z\in \{-1, 1\}^{2N}$ is drawn from the forrelation distribution, $x\sim U_{2N}$ is drawn uniformly and $y:=x\cdot z$. Alice gets $x$ and Bob gets $y$. Given any deterministic communication protocol $C:\{-1, 1\}^{2N}\times \{-1, 1\}^{2N}\rightarrow \{-1, 1\}$ of cost $c\ge 1$, its expectation when the inputs are drawn from this distribution is close to when the inputs are drawn from the uniform distribution. That is,  \[ \left|\underset{\substack{x\sim U_{2N}\\z\sim \mathcal{D}}}{\E}\left[C(x,x\cdot z)\right]- \underset{x,y\sim U_{2N}}{\E}\left[C(x,y)\right] \right|\le  O\left( \frac{c^2}{N^{1/2}}\right) \] In other words, no deterministic protocol of cost $o(N^{1/4})$ has considerable advantage in distinguishing the above distribution from the uniform distribution. \end{theorem}

\begin{definition} [The Forrelation Problem] 
Alice is given $x\in \{-1, 1\}^{2N}$ and Bob is given $y\in\{-1, 1\}^{2N}$. Their goal is to compute the partial boolean function $F$ defined as follows.
\[F(x,y)=\begin{cases} -1 & \text{ if } forr(x\cdot y ) \ge \epsilon/4 \\ 1 & \text{ if } forr(x\cdot y) \le \epsilon/8 \end{cases} \]  
\end{definition}

\begin{theorem} \label{theorem2} The forrelation problem can be solved in the quantum simultaneous with entanglement model with $O(\log^3 N)$ bits of communication, when Alice and Bob are given access to $O(\log^3N)$ bits of shared entanglement. Moreover, the protocol is efficient, as it can be implemented by a $O(\log^3 N)$ size quantum circuit with oracle access to inputs.
\end{theorem}
\begin{theorem} \label{theorem3} The randomized bounded-error interactive communication cost of the forrelation problem is $\tilde{\Omega}(N^{\frac{1}{4}})$.
\end{theorem}

\section{Proof of Theorem \ref{theorem2}: Quantum Upper Bound}

Our protocol will be based on three standard subroutines described in Figures 3, 4 and 5. The first is the {\it Swap} test between vectors $\ket{\phi}$ and $\ket{\psi}$, which takes as input the state $\frac{1}{\sqrt{2}}\ket{0} \ket{\phi} +\frac{1}{\sqrt{2}}\ket{1} \ket{\psi}$ and outputs $\ket{1}$ with probability $\frac{1+\langle \phi \mid \psi \rangle}{2}$ and $\ket{0}$ with the remaining probability. This can be implemented by applying a Hadamard on the first bit and then measuring it and negating the outcome. The probability associated with the outcome $\ket{1}$ is precisely $\frac{\|\ket{\phi}+\ket{\psi}\|^2}{4}=\frac{1+\langle \phi \mid \psi\rangle}{2}$. The second subroutine is a controlled erase/entangle operator $E$ which exchanges the basis states $\ket{i}\ket{i}$ and $\ket{i}\ket{0}$ for every $i\in [N]$. Its action on other states can be arbitrary.  It can be implemented as follows. For each register $j\in [\log N]$, negate the contents of the $(\log N + j)$-th register, controlled on the contents of the $j$-th register.  The third subroutine is a \textsc{Controlled Hadamard} operator. This is a two-qubit operator which applies $H$ on the second register if the content of the first register is $\ket{1}$ and otherwise does nothing. It can be implemented as shown in Figure 5.

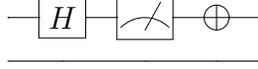
\begin{figure}[h!]
\centering
\mbox{ 
\Qcircuit @C=1em @R=.7em {
& \gate{H}  & \meter & \targ &\qw \\
& \qw & \qw & \qw &\qw
}}
\caption{The $Swap$ test}
\end{figure}
 
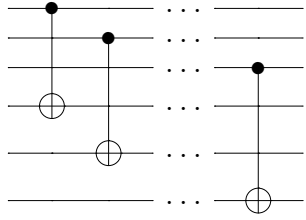
\begin{figure}[h!]
\centering
\mbox{ 
\Qcircuit @C=1em @R=.7em {
& \ctrl{3} & \qw & \qw & \ldots && \qw & \qw\\
& \qw & \ctrl{3} & \qw & \ldots && \qw &\qw \\
& \qw & \qw & \qw & \ldots&& \ctrl{3} &\qw\\
&  \targ & \qw & \qw & \ldots && \qw &\qw\\
&  \qw & \targ& \qw & \ldots&& \qw  &\qw \\
&    \qw & \qw& \qw & \ldots && \targ  &\qw\\
}
}
\caption{The $E$ operator}
\end{figure}

\begin{figure}[h!]
\centering
\mbox{ 
\Qcircuit @C=1em @R=.7em {
&\ctrl{1} & \qw &    & & \qw & \qw & \ctrl{1} & \qw & \qw  &\qw \\
&\gate{H} & \qw & = & & \gate{H}  & \gate{R_{\pi/8}} & \targ & \gate{H} & \gate{R_{\pi/8}} & \qw \\
}}
\caption{The \textsc{Controlled Hadamard} operator}
\end{figure}
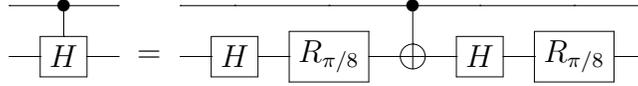

\subsection{Quantum Communication Protocol for Forrelation}
Let $m=c \log^3 (2N)$ for some large enough constant $c$. Let $M=2^m$. We first observe the following identity.
\[ \frac{1}{\sqrt{M}}\sum_{i=1}^{M} \ket{i}^A\ket{i}^B= \left( \frac{1}{\sqrt{2N}} \sum_{i=1}^{2N}\ket{i}^A\ket{i}^B \right)^{\otimes c\log^2 (2N)}\]
Henceforth, we will assume that Alice and Bob have  $c\log^2(2N) $ independent copies of the state $\frac{1}{\sqrt{2N}} \sum_{i=1}^{2N}\ket{i}^A\ket{i}^B$. Consider the following protocol based on the algorithm for forrelation by Aaronson and Ambainis~\cite{aaronsonambainis} and by Raz and Tal~\cite{raztal}.

 \begin{enumerate}[(1.)]
\item Let $x=(x_1,x_2)$ and $y=(y_1,y_2)$ be Alice's and Bob's inputs respectively, where $x_1,x_2,y_1,y_2 \in \{-1, 1\}^N$. Recall that Alice and Bob are given $c\log^2 (2N)$ copies of the following maximally entangled state.
\[ \frac{1}{\sqrt{2N}}\sum_{i=1}^{2N}\ket{i}^A\ket{i}^B= \frac{1}{\sqrt{2N}} \sum_{i=1}^N \ket{0i}^A\ket{0i}^B +\frac{1}{\sqrt{2N}}\sum_{i=1}^N \ket{1i}^A\ket{1i}^B\]
For each copy, Alice (respectively Bob) applies the oracle to her input $O_x$ (respectively $O_y$) to create the state
\[\ket{\gamma}:= \frac{1}{\sqrt{2N}} \sum_{i=1}^N \ket{0i}^A\ket{0i}^Bx_1(i)y_1(i) +\frac{1}{\sqrt{2N}}\sum_{i=1}^N \ket{1i}^A\ket{1i}^Bx_2(i)y_2(i)\]
Alice and Bob simultaneously send all their copies of this state to the referee.
\item For each copy, the referee uses the $E$ operator to create the state 
\[ \frac{1}{\sqrt{2}}\Big(\ket{0}\ket{\phi}+\ket{1}\ket{\psi}\Big) \ket{0^{\log N+1}}\] 
\[\text{where }\quad  \ket{\phi}=\frac{1}{\sqrt{N}} \sum_{i=1}^N \ket{i} x_1(i)y_1(i) \quad\text{and}\quad \ket{\psi}=  \frac{1}{\sqrt{N}} \sum_{i=1}^N \ket{i}x_2(i)y_2(i) \]
 Ignoring the last few blank registers, the referee has $\frac{1}{\sqrt{2}}\Big(\ket{0}\ket{\phi}+\ket{1}\ket{\psi}\Big)$. The referee first negates the content of the first register. He then performs a series of controlled Hadamard operators where the control is always on the first register and the target registers vary from $i=2$ to $\log N+1$. He thus obtains:
\[ \frac{1}{\sqrt{2}}\ket{1}\otimes H_N(\ket{\phi}) +  \frac{1}{\sqrt{2}}\ket{0}\otimes \ket{\psi}\]
This allows the referee to perform a $Swap$ test between $H_N(\ket{\phi})$ and $\ket{\psi}$. He appends the output of the swap test to an auxilliary register.
\item The referee returns 1 if the fraction of $1$ entries in the registers exceeds $\frac{1}{2}+ \frac{3}{32}\cdot \epsilon$ and $-1$ otherwise.
\end{enumerate}

\subsection{Correctness of the Quantum Protocol}

The expected fraction of 1 entries of the swap test between $H_N(\phi)$ and $\psi$ is
\[ \frac{1}{2}+\frac{1}{2} \left\langle  \frac{1}{\sqrt{N}} H_N\Big(\sum_i \ket{i} x_1(i)y_1(i) \Big) \Bigg| \frac{1}{\sqrt{N}} \sum_i \ket{i} x_2(i)y_2(i)  \right\rangle = \frac{1}{2} + \frac{1}{2} forr(x\cdot y)\]
The promise on the inputs is that this quantity is at least $\frac{1}{2}+\frac{\epsilon}{8}$ for \textsc{yes} instances, while it is at most $\frac{1}{2}+\frac{\epsilon}{16}$ for \textsc{no} instances. A simple application of the additive Chernoff bound implies that if $\mu$ is the random variable describing the average of $O(1/\epsilon^2)=c\log^2(2N)$ independent trials of the test, then, for a large enough constant $c$, with probability at least 2/3, the random variable is within $\epsilon/32$ of its mean. This means that for \textsc{yes} instances, the fraction of 1 entries in the referees register is greater than $\frac{1}{2}+\frac{\epsilon}{8}-\frac{\epsilon}{32}\ge \frac{1}{2}+\frac{3\epsilon}{32}$ with high probability, while for \textsc{no} instances, it is less than $\frac{1}{2}+\frac{\epsilon}{16}+\frac{\epsilon}{32}\le \frac{1}{2}+\frac{3\epsilon}{32}$ with high probability. This proves the correctness of the protocol.

\subsection{Quantum Circuit for Forrelation}

The above protocol can be described by a quantum circuit of small size (see Figure 6). We first remark that the subroutines $Swap$, the controlled Hadamard and $E$ are efficient, since the first two involve $O(1)$ single-bit operations while the $E$ operator involves $\log N$ controlled not operators. The initial entangled state $\frac{1}{\sqrt{M}} \sum_{i=1}^M \ket{i}^A\ket{i}^B$ can be created by applying a Hadamard gate on the first $\log M$ registers and then applying the erase operator on the registers $1,\ldots,2\log M$. Step (1.) requires one parallel oracle query to $O_x$ and to $O_y$ for each of the $c\log^2(2N)$ copies. Step (2.) involves a single application of the $E$ operator, a negation operator, $O(\log N)$ applications of the controlled Hadamard and one $Swap$ test, for each of the $c\log^2(2 N)$ copies. The entire circuit is thus composed of $O(\log^3N)$ operators.
 
\begin{figure}
\centering
\mbox{ 
\Qcircuit @C=1em @R=.7em {
& \ket{0} &  & \gate{H} &  \multigate{7}{E}&  \qw & \multigate{3}{\mathcal{O_A}} & \qw & \multigate{7}{E}  & \qw & \targ &\ctrl{1}  &\ctrl{2} &\qw & \ldots&& \ctrl{3} & \gate{Swap} & \qw  \\ 
& \ket{0}  &  & \gate{H}& \ghost{E} &　\qw  & \ghost{\mathcal{O_A}} & \qw  & \ghost{E} & \qw &\qw &  \gate{H} & \qw & \qw & \ldots&&\qw &\qw &    \\
& \ket{0}  &  & \gate{H}&  \ghost{E} &　\qw  & \ghost{\mathcal{O_A}} & \qw  & \ghost{E} & \qw &  \qw &\qw &  \gate{H} &\qw & \ldots&&\qw &\qw &    \\
& \ket{0}  & & \gate{H} & \ghost{E} &　 \qw & \ghost{\mathcal{O_A}} & \qw & \ghost{E} & \qw  &  \qw &\qw &\qw &  \qw&  \ldots&&\gate{H} &\qw \\
& \ket{0} &  & \qw& \ghost{E} &　\qw & \multigate{3}{\mathcal{O_B}} & \qw & \ghost{E} & \qw \\ 
& \ket{0}  &  &  \qw& \ghost{E} &　\qw& \ghost{\mathcal{O_B}} & \qw  & \ghost{E} & \qw  \\
& \ket{0} &  &  \qw& \ghost{E} &　\qw& \ghost{\mathcal{O_B}} & \qw  & \ghost{E} & \qw  \\
& \ket{0} &   &\qw  & \ghost{E} &　\qw&  \ghost{\mathcal{O_A}} & \qw & \ghost{E} & \qw \\
}
}
\caption{Circuit describing one component of the quantum protocol. The output of this component is the outcome of the $Swap$ test. The final circuit is obtained by taking a threshold of the outputs of $O(\log^2N)$ copies of this circuit.}
\end{figure}
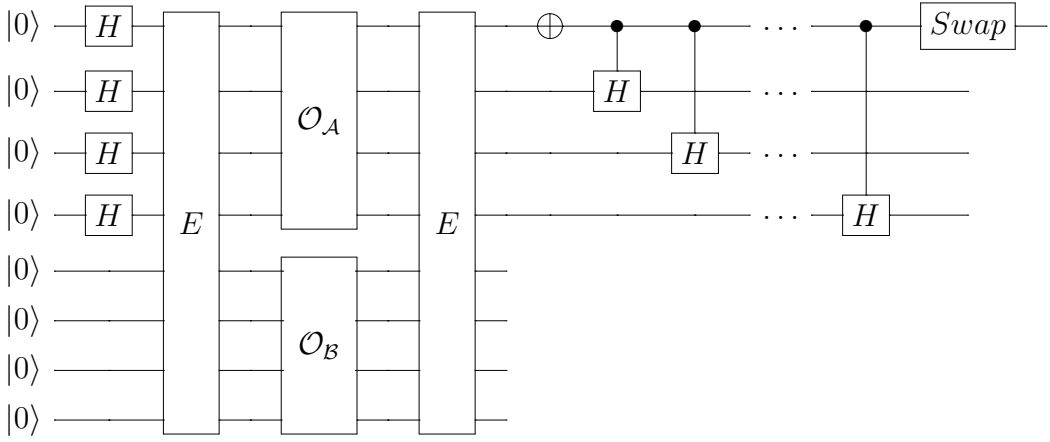

\section{Proof of Theorem \ref{theorem1}: Distributional Lower Bound}

Let $C:\{-1, 1\}^{2N}\times \{-1, 1\}^{2N}\rightarrow \{-1, 1\}$ be any deterministic protocol of cost at most $c$. Let $ D:\{-1, 1\}^{2N}\times \{-1, 1\}^{2N}\rightarrow \{-1, 1\}$ be defined as follows. For $x,z\in\{-1,1\}^{2N}$,
\[ D(x,z):=C(x,x\cdot z)\]
We will also use $D(x,z)$ to refer to its mulilinear extension. Note that our goal is to show that the function $ \underset{x\sim U_{2N}}{\E} \left[D(x, z)\right]$ of $z$ is fooled by $\mathcal{D}$. Towards this, we will prove that it is fooled by $p \mathcal{G}$ for small $p$. This approach was first used in ~\cite{chhl} and is analogous to Claim 7.2 in ~\cite{raztal}. 

\begin{lemma} \label{mainlemma1} Let $p\le \frac{1}{2N}$ and let $C(x,y)$ be any deterministic protocol of cost $c\ge1$ for the forrelation problem. As before, let $D(x, z):\mathbb{R}^{2N}\times \mathbb{R}^{2N}\rightarrow\mathbb{R}$ refer to the multilinear extension of $C(x,x\cdot z)$. Let $P\in [-p,p]^{2N}$. Then,
\[ \left| \underset{\substack{z\sim P\cdot \mathcal{G}\\x\sim U_{2N}}}{\E} \left[D(x, z)\right] - \underset{z,x\sim U_{2N}}{\E}\left[D(x,z)\right]  \right|  \le \frac{  120\epsilon c^2 p^2 }{\sqrt{N}} +p^4N^3 \]
\end{lemma}

\begin{corollary} \label{maincorollary1} Under the same hypothesis as in Lemma \ref{mainlemma1},
\[ \Big| \E_{z\sim P\cdot \mathcal{G}} \left[D(0, z)\right] - D(0,0)  \Big|  \le \frac{  120\epsilon c^2 p^2 }{\sqrt{N}} +p^4N^3 \]
\end{corollary}

\begin{proof}[Proof of Corollary \ref{maincorollary1} from Lemma \ref{mainlemma1}]
Since $D(x,z)$ is a multilinear polynomial, for all $z\in \mathbb{R}^{2N}$, we have $\E_{x\sim U_{2N}}\left[D(x,z)\right]=D(0,z)$. This implies that 
\[ \underset{\substack{z\sim P\cdot \mathcal{G}\\x\sim U_{2N}}}{\E} \left[D(x, z) \right]= \E_{z\sim P\cdot \mathcal{G}} \left[D(0, z)\right]\]
We also have $ \E_{z\sim U_{2N}}\left[D(0,z)\right]=D(0,0) $. This implies that 
\[ \underset{z,x\sim U_{2N}}{\E}\left[D(x,z)\right] = D(0,0)  \]
The proof of Corollary \ref{maincorollary1} follows from the above two equalities and Lemma \ref{mainlemma1}.
 \end{proof}
 
\begin{proof}[Proof of Lemma \ref{mainlemma1}] We begin by observing some properties of the distribution $P\cdot \mathcal{G}$. The sample $z\sim P\cdot \mathcal{G}$ is obtained by scaling the $i$-th coordinate of $z'\sim\mathcal{G}$ by $P_i$ for each $i\in[2N]$. This implies that for all $S\subseteq [2N],$
\begin{equation} \label{eqnnum1} \underset{z\sim P\cdot \mathcal{G}}{\E}\left[\chi_S(z)\right]=\left(\prod_{i\in S}P_i\right) \underset{z\sim \mathcal{G}}{\E}\left[\chi_S(z)\right] \end{equation}
Part (2.) of Claim \ref{claim4} implies that the odd moments of $\mathcal{G}$ are zero. Equation (\ref{eqnnum1}) implies that this is also true for $P\cdot \mathcal{G}$. That is, for all $S\subseteq [2N]$, 
\begin{equation} \label{eqnnum2} |S| \text{ is odd } \implies \underset{z\sim P\cdot \mathcal{G}}{\E}\left[\chi_S(z)\right]=0 \end{equation}
Part (3.) of Claim \ref{claim4} implies that for $S\subseteq [2N], |S|=2k$, the $S$-th moment $\mathbb{E}_{z\sim \mathcal{G}}\chi_S(z)$ is at most $\epsilon^k k! N^{-k/2}$ in magnitude. Along with equation (\ref{eqnnum1}), this implies that for $k\in \mathbb{N}$,
\begin{equation}\label{eqnnum3} |S|=2k\implies \left| \underset{z\sim P\cdot \mathcal{G}}{ \E}\left[{\chi_S(z)}\right]\right| \le \left(\prod_{i\in S}P_i \right) \epsilon^k k! N^{-k/2}  \le p^{2k}\epsilon^k k! N^{-k/2} \end{equation}

We now proceed with the proof of the lemma. Let
\[\Delta:=\left| \underset{\substack{z\sim P\cdot \mathcal{G}\\x\sim U_{2N}}}{\E} \left[D(x, z)\right] - \underset{z,x\sim U_{2N}}{\E}\left[D(x,z) \right] \right| \]
Note that this is the quantity we wish to bound in the lemma. For ease of notation, let $H:\{-1,1\}^{2N}\rightarrow [-1,1]$ be defined at every point $z\in \{-1,1\}^{2N}$ by
\[ H(z):=\underset{x\sim U_{2N}}{\E} \left[D(x, z)\right] \]
We identify $H(z)$ with its multilinear extension. Note that by uniqueness of multilinear extensions, the above equality holds even for $z\in \mathbb{R}^{2N}$. This implies that
\[ \underset{\substack{z\sim P\cdot \mathcal{G}\\x\sim U_{2N}}}{\E} \left[D(x, z)\right]= \underset{z\sim P\cdot \mathcal{G}}{\mathbb{E}} \left[H(z)\right] \quad \text{ and }\quad  \underset{z,x\sim U_{2N}}{\E} [D(x, z)]=  \underset{z\sim U_{2N}}{\E}  [H(z)] \]
This, along with the definition of $\Delta$ implies that
\[ \Delta=\left| \underset{z\sim P\cdot \mathcal{G}}{\E} [H(z)]- \underset{z\sim U_{2N}}{\E} [H(z)]\right| \]
Note that $H(z)=\sum_S \widehat{H}(S) \chi_S(z)$ for all $z\in \mathbb{R}^{2N}$. This implies that for all distributions $\mathcal{Z}$ on $\mathbb{R}^{2N}$, we have $ \underset{z\sim \mathcal{Z}}{\E} [H(z)] = \sum_S \widehat{H}(S) \underset{z\sim \mathcal{Z}}{\E} [\chi_S(z)]$. This implies that
\[ \Delta= \left| \sum_{S\subseteq [2N]} \widehat{H}(S) \left( \underset{z\sim P\cdot \mathcal{G}}{\E} [\chi_S(z)] -  \underset{z\sim U_{2N}}{\E} [\chi_S(z) ]\right) \right| \]
For any probability distribution, the moment corresponding to the empty set is 1 by definition. For all non empty sets $S$, we have $\underset{z\sim U_{2N}}{\E} [\chi_S(z)]=0$. Using this fact in the above equality, along with the triangle inequality, we have
\[ \Delta = \left| \sum_{\emptyset\neq S\subseteq [2N]}  \widehat{H}(S) \underset{z\sim P\cdot  \mathcal{G}}{\E}[\chi_S(z)]  \right|   \le \sum_{\emptyset\neq S \subseteq [2N]} \left| \widehat{H}(S)\right| \left|  \underset{z\sim P\cdot \mathcal{G}}{\E}[ \chi_S(z)]  \right|  \] 
We use the bounds from (\ref{eqnnum2}) and (\ref{eqnnum3}) on the moments of $P\cdot \mathcal{G}$ to derive the following.
\begin{align*}\begin{split}
 \Delta&\le  \underset{\substack{|S|=2k\\k \ge 1}}{\sum}\left| \widehat{H}(S)\right|  p^{2k} \epsilon^k k! N^{-k/2} \\
&= \underset{k \ge 1}{\sum} L_{2k}(H) p^{2k}  \epsilon^k k! N^{-k/2} \\
\end{split}\end{align*}
We upper bound $L_{2k}(H)$ by ${2N\choose 2k}$ when $k\ge2$. This implies that
\begin{align*}\begin{split}
\Delta&\le L_{2}(H) \frac{\epsilon p^2 }{\sqrt{N}} +  \sum_{k \ge 2} {2N\choose 2k}  p^{2k}\epsilon^k k!N^{-k/2} \\
&\le L_{2}(H) \frac{\epsilon p^2 }{\sqrt{N}} +  \sum_{k \ge 2} \frac{2^{2k}N^{2k}}{(2k)!} p^{2k}\epsilon^k k!N^{-k/2} \\
&\le L_{2}(H) \frac{\epsilon p^2 }{\sqrt{N}}  +  \sum_{k \ge 2}N^{3k/2}p^{2k}4^k\epsilon^k \\
\end{split}\end{align*}
In the summation $\sum_{k \ge 2}N^{3k/2}p^{2k}4^k\epsilon^k$, we see that every successive term is smaller than the previous by a factor of at least $1/4$. This is because the assumption $p\le \frac{1}{2N}$ implies that $p^2N^{3/2}\le p^2N^2\le \frac{1}{4}$ and because $4\epsilon \le 1$. Thus, we can bound this summation by twice the first term, which is $16p^4 N^3\epsilon^2.$ This implies that
\[\Delta \le L_2(H) \frac{\epsilon p^2 }{\sqrt{N}}  +32p^4N^3\epsilon^2 \]
Since $\epsilon = \frac{1}{50\ln N} \le \frac{1}{32}$, we may bound $32p^4N^3\epsilon^2$ by $p^4N^3$. This implies that
\[\Delta \le   L_2(H)\frac{\epsilon p^2 }{\sqrt{N}} +p^4N^3 \]
The following claim provides a bound on $L_2(H)$.
\begin{claim} \label{weightbound} Let $C(x,y):\{-1,1\}^{2N}\times \{-1,1\}^{2N}\rightarrow \{-1,1\}$ be any deterministic protocol of cost $c\ge1$, let $D(x, z):\mathbb{R}^{2N}\times \mathbb{R}^{2N}\rightarrow\mathbb{R}$ refer to the unique multilinear extension of $C(x,x\cdot z)$ and $H:\mathbb{R}^{2N}\rightarrow \mathbb{R}$ be defined by $H(z)=\mathbb{E}_{x\sim U_{2N}} D(x,z)$. Then,
\[ L_2(H)\le 120c^2\]
\end{claim} 
This claim along with the preceding inequality implies that
\[\Delta \le 120c^2 \frac{\epsilon p^2}{\sqrt{N}} + p^4N^3  \]
This completes the proof of Lemma \ref{mainlemma1}.
\end{proof}

\begin{proof}[Proof of Claim \ref{weightbound}] In order to bound the level-2 Fourier mass of $H$, we will use the following lemma. Its statement and proof appear as `Level-$k$ Inequalities' on Page 259 of `Analysis of Boolean Functions' \cite{odonnell}.
\begin{lemma}[Level-$k$ Inequalities] \label{levelkinequality}  Let $F:\{-1,1\}^n\rightarrow \{0,1\}$ have mean $\E[F]=\alpha$ and let $k\in \mathbb{N}$ be at most $2\ln(1/\alpha)$. Then,
\[ \sum_{|S|=k} \left(\widehat{F}(S)\right)^2\le\alpha^2 \left(\frac{2e}{k}\ln(1/\alpha)\right)^k \]
\end{lemma}

We now show the desired bound on $L_2(H)$. Since $C$ is a deterministic protocol of cost at most $c$, it induces a partition of the input space $\{-1,1\}^{2N}\times \{-1,1\}^{2N}$ into at most $2^c$  rectangles. Let $\mathcal{P}$ be this partition and let $A\times B$ index rectangles in $\mathcal{P}$, where $A$ (respectively $B$) is the set of Alice's (respectively Bob's) inputs compatible with the rectangle. Let $C(A\times B)\in\{-1,1\}$ be the output of the protocol on inputs from a rectangle $A\times B\in \mathcal{P}$. For all $x,y\in \{-1, 1\}^{2N}$, we have
\[ C(x,y) = \underset{A\times B \in \mathcal{P}}{\sum} C(A\times B) \mathbbm{1}_{A}(x)  \mathbbm{1}_B(y)\]
By definition, $D(x,z)=C(x,x\cdot z)$. This implies that
\[ D(x,z) = \underset{A\times B \in \mathcal{P}}{\sum} C(A\times B) \mathbbm{1}_{A}(x)  \mathbbm{1}_B(x\cdot z)\]
Taking an expectation over $x\sim U_{2N}$ of the above identity implies that
\[ H(z)\triangleq \underset{x\sim U_{2N}}{\E} [D(x,z)]= \sum_{A\times B\in \mathcal{P}}C(A\times B) \big(\mathbbm{1}_A*\mathbbm{1}_B\big)(z) \]
This implies that for any $S\subseteq [n]$, we have
\[\widehat{H}(S)= \sum_{A\times B\in \mathcal{P}}C(A\times B) \widehat{\mathbbm{1}_A*\mathbbm{1}_B}(S)=\sum_{A\times B\in \mathcal{P}}C(A\times B) \widehat{\mathbbm{1}_A}(S)\widehat{\mathbbm{1}_B}(S)\]
We thus obtain
\begin{align*}\begin{split}
L_2(H)&=\sum_{|S|=2}\left| \widehat{H}(S) \right| \\
&= \sum_{|S|=2}\left| \underset{A\times B\in \mathcal{P}}{\sum} C(A\times B)\widehat{\mathbbm{1}_A}(S)\widehat{\mathbbm{1}_B}(S)\right| \\ 
&\le \underset{A\times B\in \mathcal{P}}{\sum} \sum_{|S|=2}  |\widehat{\mathbbm{1}_A}(S)| |\widehat{\mathbbm{1}_B}(S)|\\
\end{split}\end{align*}
We apply Cauchy Schwarz to the term $\sum_{|S|=2}  |\widehat{\mathbbm{1}_A}(S)| |\widehat{\mathbbm{1}_B}(S)|$ to obtain
\[ L_2(H) \le  \underset{A\times B\in \mathcal{P}}{\sum} \Big( \sum_{|S|=2} \widehat{\mathbbm{1}_A}(S)^2 \Big)^{1/2}\Big( \sum_{|S|=2} \widehat{\mathbbm{1}_B}(S)^2 \Big)^{1/2} \]
For ease of notation, let $\mu(A)=\frac{|A|}{2^{2N}}$ denote the measure of a set $A\subseteq \{-1, 1\}^{2N}$ under $U_{2N}$. We first ensure that for each rectangle $A\times B\in \mathcal{P}$, we have $\mu(A)\le \frac{1}{e}$ and $\mu(B)\le \frac{1}{e}$. We may do this by adding 2 extra bits of communication for each player. For $k=2$, we have $k= 2\ln(e)\le 2\ln\frac{1}{\mu(A)}$ and $k\le 2\ln\frac{1}{\mu(B)}$. We apply Lemma \ref{levelkinequality} on the indicator functions $\mathbbm{1}_A$ and $\mathbbm{1}_B$ for $k=2$ to obtain
\[\sum_{|S|=2} \left(\widehat{\mathbbm{1}_A}(S)\right)^2 \le \mu(A)^2\Big(e \ln(1/\mu(A))\Big)^2 \quad\text{and}\quad \sum_{|S|=2} \left(\widehat{\mathbbm{1}_B}(S)\right)^2 \le \mu(B)^2\Big(e \ln(1/\mu(B))\Big)^2 \] 
Substituting this in the bound for $L_2(H)$, we have
\[ L_2(H) \le e^2 \underset{A\times B\in \mathcal{P}}{\sum} \mu(A)\mu(B)\ln\frac{1}{\mu(A)}\ln\frac{1}{\mu(B)} \]
Let $\Delta:= e^2\underset{A\times B\in \mathcal{P}}{\sum} \mu(A)\mu(B)\ln\frac{1}{\mu(A)}\ln\frac{1}{\mu(B)}$ be the expression in the R.H.S. of the above. Note that it suffices to upper bound $\Delta$. Consider the case when $\mathcal{P}$ consists of $2^c$ rectangles $A\times B$, each of which satisfies $\mu(A)=\mu(B)=\frac{1}{2^{c/2}}$. In this case, $\Delta$ evaluates to $e^2\sum_{A\times B\in \mathcal{P}} \frac{1}{2^c}( \frac{c\ln2}{2})^2 =O(c^2)$. This proves the lemma in this special case. A similar bound holds for the general case and the proof follows from a concavity argument that we describe now. 

\begin{figure}
\centering
\begin{tikzpicture}[scale=0.7]
    \begin{axis}[
            axis lines=middle,
            xmin=0,xmax=1,ymin=0,ymax=0.6,
            xlabel=$x$,
            ylabel=$y$,
    y label style={at={(axis description cs:0.05,1.0)}},
    x label style={at={(axis description cs:1.0,0.05)}},
            ]
      \addplot[domain=0:1,thick,samples=100] {x*ln(1/x)^2};
    \end{axis}
  \end{tikzpicture}
\caption{Plot of the function $y=x\left(\ln\frac{1}{x}\right)^2$}
\end{figure}
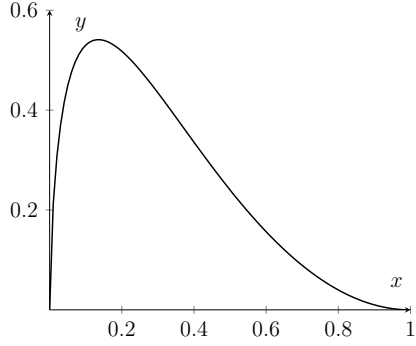

Since $\mu(A),\mu(B)\le 1$, we have the following inequality.
\begin{align*}\begin{split}
\Delta &\triangleq e^2\underset{A\times B\in \mathcal{P}}{\sum} \mu(A)\mu(B)\ln\frac{1}{\mu(A)}\ln\frac{1}{\mu(B)}\\
&\le e^2\underset{A\times B\in \mathcal{P}}{\sum} \mu(A)\mu(B)\ln\frac{1}{\mu(A)\mu(B)}\ln\frac{1}{\mu(A)\mu(B)}\\
&=e^2\underset{A\times B\in \mathcal{P}}{\sum} \mu(A\times B)\left(\ln\frac{1}{\mu(A\times B)}\right)^2
\end{split}\end{align*}

Let $f:[0,\infty)\rightarrow \mathbb{R}$ be defined by $f(p):=p\ln (1/p)^2$. A small calculation shows that $f$ is a concave function in the interval $[0,0.3]$ (see Figure 7). Let $\alpha_i\in [0,0.3]$ for $i\in [k]$. Jensen's inequality applied to $f$ states that for $i\sim [k]$ drawn uniformly at random, we have $\E_i [f(\alpha_i)]\le f(\E_i [\alpha_i])$. This implies that
\[ \sum_{i=1}^k \alpha_i  \ln(1/\alpha_i)^2 \le \left(\sum_{i=1}^k\alpha_i\right)\ln\left(\frac{k}{\sum_{i=1}^k \alpha_i}\right)^2 \]
We apply this inequality to the terms in $\Delta$ by substituting $\alpha_i$ with $\mu(A\times B)$. We may do this since the assumption that $\mu(A),\mu(B)\le \frac{1}{e}$ implies that $\mu(A\times B)\le \frac{1}{e^2}\le 0.3$. This implies that
\[ \Delta \le e^2 \left( \sum_{A\times B\in \mathcal{P}}\mu(A\times B)\right) \ln\left( \frac{2^{c+4}}{\sum_{A\times B\in \mathcal{P}} \mu(A\times B)} \right)^2  \] 
Since $\sum_{A\times B\in \mathcal{P}}\mu(A\times B)=1$, we have
\[\Delta \le e^2(c+4)^2 (\ln 2)^2 \le 120c^2  \]
This completes the proof of Claim \ref{weightbound}.\end{proof}

We now show that an analogue of Lemma \ref{mainlemma1} holds for restricted protocols, similarly to Claim 7.3 in ~\cite{raztal}.
\begin{lemma} \label{mainlemma2} Let $p\le \frac{1}{4N}$ and $C(x,y)$ be any deterministic protocol of cost $c\ge 1$ for the forrelation problem. As before, let $D(x, z):\mathbb{R}^{2N}\times \mathbb{R}^{2N}\rightarrow \mathbb{R}$ refer to the multilinear extension of $C(x,x\cdot z)$. Let $z_0\in [-1/2,1/2]^{2N}$.  Then,
\[ \left| \underset{\substack{z\sim p \mathcal{G}\\x\sim U_{2N}}}{\E} [D(x, z_0+  z)] - \underset{{z,x\sim U_{2N}}}{ \E}[D(x,z_0+ z)] \right| \le  \frac{120\epsilon c^2 (2p)^2}{\sqrt{N}} +(2p)^4N^3  \]
\end{lemma}

\begin{corollary} \label{maincorollary2} Under the same hypothesis as in Lemma \ref{mainlemma2},
\[ \Big| \E_{z\sim p \mathcal{G}}[D(0,z_0+ z)] - D(0,z_0)  \Big|  \le \frac{  120\epsilon c^2 (2p)^2 }{\sqrt{N}} +(2p)^4N^3 \]
\end{corollary}

\begin{proof}[Proof of Corollary \ref{maincorollary2} from Lemma \ref{mainlemma2}]
Since $D(x,z)$ is a multilinear polynomial, for all $z\in \mathbb{R}^{2N}$, we have $\E_{x\sim U_{2N}}[D(x,z)]=D(0,z)$. This implies that for all $z_0\in \mathbb{R}^{2N}$,
\[ \underset{\substack{z\sim p \mathcal{G}\\x\sim U_{2N}}}{\E}[D(x, z_0+z)] = \E_{z\sim p \mathcal{G}}[ D(0, z_0+z)]\]
For all $z_0\in \mathbb{R}^{2N}$, since $ \E_{z\sim U_{2N}}[D(0,z_0+z)]=D(0,z_0) $, we have 
\[ \underset{z,x\sim U_{2N}}{\E}[D(x,z_0+z)] = D(0,z_0)  \]
The proof of Corollary \ref{maincorollary2} follows from the above two equalities and Lemma \ref{mainlemma2}.
 \end{proof}
 
\begin{proof}[Proof of Lemma \ref{mainlemma2}] Similarly to the approach of~\cite{chhl,raztal}, we will express $D(x,z_0+z)$ as the average output of restricted protocols $(C\circ\rho)(x,x\cdot z)$, on which we can use Lemma \ref{mainlemma1} to derive the result. These restricted protocols roughly correspond to Alice and Bob fixing a common subset $I\subseteq[2N]$ of their inputs in a predetermined way and then running the original protocol. We formalize this now. 

A restriction $\rho$ of $\mathbb{R}^{2N}$ is an element of $\{-1, 1, *\}^{2N}$. It defines an action $\rho:\mathbb{R}^{2N}\rightarrow \mathbb{R}^{2N}$ in the following natural way. For any $z\in \mathbb{R}^{2N}$ and $i\in [2N]$,
\[ (\rho(z))(i): = \begin{cases} \rho(i) &  \text{ if } \rho(i) \in \{-1, 1\} \\
 z(i) & \text{ otherwise } \end{cases}\] 

Let $sign:(\mathbb{R}\setminus 0)\rightarrow \{-1,1\}$ be the function which maps real numbers to their sign. Given $z_0\in [-1/2,1/2]^{2N}$, let $R_{z_0}$ be a distribution over restrictions of $\mathbb{R}^{2N}$ defined as follows.  For each $i\in [2N]$, independently, set\footnote{If $z_0(i)$ is zero, then $\rho(i)=*$ with probability 1.}: 
\[\rho(i) := \begin{cases} sign(z_0(i)) & \text{ with probability } |z_0(i)| \\
 * &\text{ with probability } 1-|z_0(i)| \end{cases}\] 

Let $P\in \mathbb{R}^{2N}$ be such that $P_i:= \frac{1}{1-|z_0(i)|}$ for every $i\in [2N]$. Note that the assumption of $z_0\in [-1/2,1/2]^{2N}$ ensures that $P$ is a well defined element of $[1,2]^{2N}$. For any $z\in \mathbb{R}^{2N}$ and $i\in [2N]$, the expected value of the $i$th coordinate of $ \rho(z)$ when $\rho\sim R_{z_0}$ can be computed as follows.
\[ \underset{\rho\sim R_{z_0}}{\E} [(\rho(z))(i)] = |z_0(i)| sign(z_0(i))+ (1-|z_0(i)|)z(i) = z_0(i) +\frac{1}{P_i} z(i) \]
This implies that for any fixed $x,z\in \mathbb{R}^{2N}$ and $z_0\in [-1/2,1/2]^{2N}$, since $D$ is a multilinear function, we have
\[ \underset{\rho\sim R_{z_0}}{\E} \left[D(x, \rho(z))\right]= D(x,\underset{\rho\sim R_{z_0}}{\E}[\rho(z)]) =D(x, z_0+P^{-1}\cdot z) \]
Replacing $z$ with $P\cdot z$ in the above equality implies that
\[ \underset{\rho\sim R_{z_0}}{\E} [D(x, \rho(P\cdot z))] =D(x, z_0+ z) \]
This equality allows us to rewrite the L.H.S. of Lemma \ref{mainlemma2} as follows.
\begin{align*}\begin{split}
\Delta:=& \left|  \underset{\substack{z\sim p \mathcal{G},\\x\sim U_{2N}}}{\E} [D(x,z_0+ z)] -\underset{z,x\sim U_{2N}}{ \E} [D(x,z_0+ z)]  \right|  \\
=& \left| \underset{\substack{z\sim pP\cdot  \mathcal{G},\\x\sim U_{2N}}}{\E}\underset{\rho\sim R_{z_0}}{\E}  [D(x, \rho(z))] - \underset{\substack{z\sim P\cdot U_{2N},\\x\sim U_{2N}}}{ \E} \underset{\rho\sim R_{z_0}}{\E}  [D(x, \rho(z))] \right|   \\
=& \left| \underset{\rho\sim R_{z_0}}{\E} \left[ \underset{\substack{z\sim pP\cdot  \mathcal{G},\\x\sim U_{2N}}}{\E} [D(x, \rho(z))] - \underset{\substack{z\sim P\cdot U_{2N},\\x\sim U_{2N}}}{ \E} [D(x, \rho(z))] \right] \right|   \\
\end{split}\end{align*}
For a multilinear polynomial, its expectation over a product distribution depends only on the mean of that distribution. This allows us to replace the expectation of $D(x,\rho(z))$ over $z\sim P\cdot U_{2N}$ by an expectation over $z\sim U_{2N}$. We thus obtain
\begin{equation}\label{eqnnum5}
\Delta= \left| \underset{\rho\sim R_{z_0}}{\E} \left[ \underset{\substack{z\sim pP\cdot  \mathcal{G},\\x\sim U_{2N}}}{\E} [D(x, \rho(z))] - \underset{\substack{z\sim U_{2N},\\x\sim U_{2N}}}{ \E} [D(x, \rho(z))] \right] \right|     \end{equation}
For any $\rho\in \{-1,1,*\}^{2N}$ and $u\in \{-1,1\}^{2N}$, we define a substitution $\rho^u:\mathbb{R}^{2N}\rightarrow \mathbb{R}^{2N}$ obtained from $\rho$ and $u$ as follows. For any $x\in \mathbb{R}^{2N}$ and $i\in [2N]$,
\[ (\rho^u(x))(i): = \begin{cases} u(i) &  \text{ if } \rho(i) \in \{-1, 1\} \\
 x(i) & \text{ otherwise } \end{cases}\] 
This is an action on $\mathbb{R}^{2N}$ which replaces the values of coordinates specified by $\rho$, with values from $u$. For every fixed $\rho$, as we vary over $x,u\sim U_{2N}$ the distribution of $\rho^u(x)$ is exactly $U_{2N}$. This implies that for all $z\in \mathbb{R}^{2N}, \rho\in \{-1,1,*\}^{2N}$,
\[\label{equality3} \underset{x\sim U_{2N}}{\E} [ D(x,\rho(z))]= \underset{x,u\sim U_{2N}}{\E} [D( \rho^u (x),\rho(z) )]\]
Substituting this in equation (\ref{eqnnum5}), we have
\[\Delta= \left|  \underset{\rho\sim R_{z_0}}{\E}  \underset{u\sim U_{2N}}{\E}  \left[ \underset{\substack{z\sim pP\cdot  \mathcal{G},\\x\sim U_{2N}}}{\E}[D(\rho^u(x), \rho(z))] -\underset{{z,x\sim U_{2N}}}{ \E}   [D(\rho^u(x), \rho(z))] \right] \right| \]
Applying Triangle Inequality on the above, we have
\begin{equation} \label{eqnnum6}\Delta \le  \underset{\rho\sim R_{z_0}}{\E}\underset{u\sim U_{2N}}{\E}  \left| \underset{\substack{z\sim pP\cdot  \mathcal{G},\\x\sim U_{2N}}}{\E}[D(\rho^u(x), \rho(z))] -\underset{{z,x\sim U_{2N}}}{ \E}[D(\rho^u(x), \rho(z))]\right|\quad \end{equation}

Fix any $\rho\in \{-1, 1,*\}^{2N}$ and $u\in\{-1,1\}^{2N}$. For every $x,z\in \{-1, 1\}^{2N},$ we have $D(x,z)=C(x,x\cdot z)$, furthermore, $\rho^u(x),\rho(z)\in \{-1,1\}^{2N}$. This implies that for every $x,z\in \{-1, 1\}^{2N}$,
\begin{equation}\label{eqnnum7} D( \rho^u (x),\rho(z) ) = C( \rho^u (x), \rho^u (x)\cdot \rho(z))  \end{equation}
This prompts us to define a communication protocol $C\circ \rho^u$ where Alice and Bob first restrict their inputs and then run the original protocol $C$. The restriction is that for each coordinate $i\in[2N]$ with $\rho_i\in \{-1, 1\},$ Alice overwrites her input $x_i$ with $u_i$ while Bob overwrites his input $y_i$ with $\rho_i u_i$. The main property of this restricted protocol is that for all $x,z\in \{-1, 1\}^{2N}$,
\[ (C\circ \rho^u)( x,x\cdot z) =C( \rho^u (x), \rho^u (x)\cdot \rho(z)) \]
This, along with equation (\ref{eqnnum7}) implies that $D(\rho^u(x),\rho(z))$ is the unique multilinear extension of $(C\circ\rho^u)(x,x\cdot z)$. The cost of $C\circ \rho^u$ is at most that of $C$ since Alice and Bob don't need to communicate to restrict their inputs.  We now use Lemma \ref{mainlemma1} on $C\circ\rho^u$ to argue that $pP\cdot \mathcal{G}$ fools $\underset{x\sim U_{2N}}{\E}[D(\rho^u(x),\rho( z))]$. The conditions of the lemma are satisfied since $pP\in [-2p,2p]^{2N}$, $p\le \frac{1}{4N}$, and $C\circ \rho^u$ is a protocol of cost at most $c$ and whose multilinear extension is $D(\rho^u(x),\rho(z))$. The lemma implies that
\[ \left|   \underset{\substack{z\sim pP\cdot \mathcal{G},\\x\sim U_{2N}}}{\E} [D(\rho^u(x),\rho(z))] - \underset{\substack{z\sim U_{2N},\\x\sim U_{2N}}}{\E}[D(\rho^u(x),\rho(z))]   \right| \le \frac{120\epsilon c^2(2p)^2}{\sqrt{N}}+ (2p)^4N^3  \] 
Substituting this in inequality $(\ref{eqnnum6})$ completes the proof of Lemma \ref{mainlemma2}.
\end{proof}

\begin{proof}[Proof of Theorem \ref{theorem1}]
Since $D(x,z)$ is the multilinear extension of $C(x,x\cdot z)$ and since $\mathcal{D}$ and $U_{2N}$ are distributions over $\{-1,1\}^{2N}$, we have 
\[ \E_{x\sim U_{2N},z\sim \mathcal{D}}[C(x,x\cdot z)]=\E_{x\sim U_{2N},z\sim \mathcal{D}}[D(x,z)]=\E_{z\sim \mathcal{D}}[D(0,z)]\]
When $x\sim U_{2N}$ and $y\sim U_{2N}$ are independently sampled, the distribution of $(x,x\cdot y)$ is $U_{4N}$. This implies that 
\[ \E_{x,y\sim U_{2N}}[C(x,y)]=\E_{x,y\sim U_{2N}}[D(x,x\cdot y)]=D(0,0)\]
The above two equations allow us to rewrite the quantity in the L.H.S. of Theorem \ref{theorem1} as follows.
\[ \Delta :=\left|\underset{\substack{x\sim U_{2N}\\z\sim \mathcal{D}}}{\E}[C(x,x\cdot z)]-\underset{x,y\sim U_{2N}}{\E}[C(x,y)] \right| = \Big| \E_{z\sim \mathcal{D}} [D(0,z)] - D(0,0) \Big| \]

\sloppy Claim \ref{claim1} applied on the multilinear polynomial $D$ implies that $ \E_{z\sim \mathcal{D}}[D(0,z)]=\E_{z\sim \mathcal{G}}[D(0,trnc(z))]$. Substituting this in the above equality implies that
\[ \Delta=  \Big| \E_{z\sim \mathcal{G}}[D(0,trnc(z))]  - D(0,0) \Big|\]
Let $t=16N^4,p=\frac{1}{\sqrt{t}}=\frac{1}{4N^2}$. Let $z^{(1)},\ldots,z^{(t)} \sim \mathcal{G}$ be independent samples and let $Z$ refer to this collection of random variables. For $i\in [t]$, define $z^{\le (i)}:=p(z^{(1)}+\ldots +z^{(i)})$. By convention, $z^{\le (0)}:=0$. Note that for $i\in[t]$, $z^{\le(i)}$ has a Gaussian distribution with mean 0 and covariance matrix as $p^2i$ times that of $\mathcal{G}$. Thus, $z^{\le(t)}$ is sampled according to $\mathcal{G}$. Substituting this in the previous equality implies that
\[ \Delta= \left|  \E_Z[D(0,trnc(z^{\le t}))]  - D(0,0) \right|  \]
To bound the above quantity, for each $0\le  i\le t-1$, we show a bound on
\[ \Delta_i:= \Big| \E_Z[D(0,trnc(z^{\le (i+1)}))] -  \E_Z[D(0,trnc(z^{\le (i)})) ] \Big|   \]
Since $z^{\le(0)}=0$, the triangle inequality implies that $ \Delta \le \sum_{i=0}^{t-1} \Delta_i$.

Fix any $i\in\{0,\ldots,t-1\}$. We now bound $\Delta_i$. Let $E_i$ be the event that $z^{\le(i)}\notin[-1/2,1/2]^{2N}$. We first observe that $E_i$ is a low probability event. Since each $z^{\le (i)}(j)$ is distributed as $\mathcal{N}(0,p^2 i\epsilon)$, where $p^2i\le 1 $ and $\epsilon=1/(50\ln N)$, we have
\[ \p[z^{\le (i)}(j)\notin [-1/2,1/2]]\le  \p[ |\mathcal{N}(0,\epsilon)| \ge 1/2 ] \le \exp(-1/8\epsilon)\le \exp(-6\ln N)= \frac{1}{N^6} \]
Applying a Union bound over coordinates $j\in [2N]$, we have for each $0\le i \le t$,
\begin{equation}\label{eqnnum8} \p[E_i] = \p[z^{\le (i)}\notin [-1/2,1/2]^{2N}] \le 2N\frac{1}{N^6} \le \frac{2}{N^5} \end{equation}
When $E_i$ does not occur, we have $trnc(z^{\le (i)})=z^{\le (i)}\in [-1/2,1/2]^{2N}$. For every fixed value of $z^{\le(i)}$ in this range, we apply Corollary \ref{maincorollary2} with parameters $p=\frac{1}{4N^2}, z_0=z^{\le (i)}$ and $z=z^{\le (i+1)}-z^{\le(i)}=pz^{(i+1)}$. Note that the conditions in the hypothesis are satisfied since $z_0\in [-1/2,1/2]^{2N}$, $p\le 1/(4N)$ and the random variable $pz^{(i+1)}$ is distributed as $p\mathcal{G}$. The corollary implies that for every $z^{\le(i)}\in [-1/2,1/2]^{2N}$, 
\[ \Big| \E_Z  \left[D(0,z^{\le (i+1)}) \mid z^{\le(i)} \right]- \E_Z  \left[ D(0,z^{\le (i)}) \mid z^{\le(i)} \right] \Big| \le   \frac{120 \epsilon c^2(2p)^2}{N^{1/2}} +(2p)^4N^3  \]
Since $\neg E_i$ implies that $z^{\le(i)}\in [-1/2,1/2]^{2N}$, we have
\[ \Big| \E_Z  \left[D(0,z^{\le (i+1)}) \mid \neg E_i \right]- \E_Z  \left[ D(0,z^{\le (i)})\mid \neg E_i \right] \Big| \le   \frac{120 \epsilon c^2 (2p)^2}{N^{1/2}} +(2p)^4N^3  \]
We apply Claim \ref{claim3} on the multilinear polynomial $D(0,z):[-1,1]^{2N}\rightarrow [-1,1]$ with the parameters $p=\frac{1}{4N^2},  z_0=z^{\le(i)}$ and $z=z^{(i+1)}$. Note that the conditions are satisfied since $z_0\in[1/2,1/2]^{2N}$ and $p\le \frac{1}{2}$. The claim implies that
\[ \Big| \E_Z\left[D(0,z^{\le (i+1)}) \mid \neg E_i\right] - \E_Z\left[D(0,trnc(z^{\le (i+1)})) \mid \neg E_i\right]  \Big| \le \frac{8}{N^5} \]
The previous two inequalities, along with the triangle inequality, imply that
\begin{equation}\label{eqnnum9} \Big| \E_Z\left[D(0,trnc(z^{\le (i+1)})) \mid \neg E_i\right] -  \E_Z\left[D(0,z^{\le (i)}) \mid \neg E_i\right]  \Big|  \le   \frac{120 \epsilon c^2 (2p)^2}{N^{1/2}} +(2p)^4N^3  + \frac{8}{N^5}\end{equation}

Note that for every possible values of $z^{\le(i+1)}$ and $z^{\le(i)}$, the difference $D(0,trnc(z^{\le (i+1)})) -  D(0,trnc(z^{\le (i)}))$ is bounded in magnitude by 2, since $D(0,trnc(z))$ maps $\mathbb{R}^{2N}$ to $[-1,1]$. This implies that
\[   \Big| \E_Z\left[D(0,trnc(z^{\le (i+1)})) \mid E_i\right] -  \E_Z\left[D(0,trnc(z^{\le (i)})) \mid E_i\right]  \Big| \le 2   \]
Thus, we have
\begin{align*}\begin{split}
\Delta_i & \le  \p[\neg E_i] \cdot \Big| \E_Z[D(0,trnc(z^{\le (i+1)})) \mid \neg E_i] -  \E_Z[D(0,trnc(z^{\le (i)})) \mid \neg E_i]  \Big| \\
& + \p[E_i]\cdot \Big| \E_Z[D(0,trnc(z^{\le (i+1)})) \mid E_i] -  \E_Z[D(0,trnc(z^{\le (i)}) )\mid E_i]  \Big|  \\
&\le  \Big| \E_Z[D(0,trnc(z^{\le (i+1)})) \mid \neg E_i] -  \E_Z[D(0,trnc(z^{\le (i)})) \mid \neg E_i]  \Big|  + 2\p[E_i]\\
&=  \Big| \E_Z[D(0,trnc(z^{\le (i+1)})) \mid \neg E_i] -  \E_Z[D(0,z^{\le (i)}) \mid \neg E_i]  \Big|  + 2\p[E_i]\\
& \le  \frac{120\epsilon c^2(2p)^2}{N^{1/2}}  + (2p)^4N^3 + \frac{8}{N^5}  + \frac{4}{N^5}
\end{split}\end{align*} 
The equality in the fourth line follows from the fact that whenever $E_i$ does not occur, $trnc(z^{\le (i)})=z^{\le (i)}$ by definition. The last inequality follows from inequalities (\ref{eqnnum8}) and (\ref{eqnnum9}). Along with the fact that $t=\frac{1}{p^2}=16N^4$, and $\epsilon\le 1$, this implies that
\begin{align*}\begin{split}
\Delta&\le \sum_{i=0}^{t-1} \Delta_i  \\
&\le  t\Big(  \frac{120 \epsilon c^2(2p)^2}{N^{1/2}} +(2p)^4N^3  + \frac{12}{N^5}\Big)\\
&\le \frac{480 \epsilon c^2}{N^{1/2}} + 16p^2 N^3 +\frac{192}{N}\\
&=O\left( \frac{c^2}{N^{1/2}} + \frac{1}{N} \right)\\
&=O\left(\frac{c^2}{N^{1/2}}\right)
\end{split}\end{align*}
The last line follows from the assumption that $c\ge 1$. This completes the proof of Theorem \ref{theorem1}.
\end{proof}

\section{Proof of Theorem \ref{theorem3}: Randomized Lower Bound}
Let $C:\{-1,1\}^{2N}\times \{-1,1\}^{2N} \rightarrow \{-1,1\}$ be a randomized protocol for the forrelation problem with cost at most $c$ and with error at most $1/3$. Consider a randomized protocol $R:\{-1,1\}^{2N}\times \{-1,1\}^{2N}\rightarrow \{-1,1\}$ defined by repeating $C$ independently $O(\ln\ln N)$ times and taking the majority of the outputs. A simple application of Chernoff's bound implies that for \textsc{yes} instances, the majority of outputs of $O(\ln\ln N)$ independent copies of $C$, is $1$ with probability at most $\frac{\epsilon}{32}=O\left(\frac{1}{\ln N}\right)$. Similarly, for \textsc{no} instances, the majority of outputs of $O(\ln\ln N)$ independent copies of $C$, is $-1$ with probability at most $\frac{\epsilon}{32}$. Thus, $R$ solves the forrelation problem with error at most $\epsilon/32$ and is of cost $O(c\ln\ln N)$. Let $D_R$ be the distribution over deterministic protocols defined by $R$. For any $x,y\in\{-1,1\}^{2N}$, let $R(x,y)=\E_{D\sim D_R}[D(x,y)]$ denote the average output of the protocol $R$ on input $(x,y)$. Note that if $(x,y)$ is a \textsc{yes} instance, we have $R(x,y)\le -1+\epsilon/16$, and if $(x,y)$ is a \textsc{no} instance, we have $R(x,y)\ge 1-\epsilon/16$. 

Let $x=(x_1,x_2)$ and $y=(y_1,y_2)$ be Alice's and Bob's inputs to the forrelation problem respectively, where $x_1,x_2,y_1,y_2\in \{-1, 1\}^{N}$. For $i,j\in \{0,1\}^{\log N}$, let $\langle i\mid j\rangle _{\mathbb{F}_2}:=\sum_{k=1}^{\log N} i(k)j(k) \mod 2$. This denotes the inner product between $i$ and $j$ over $\mathbb{F}_2$. Recall the definition of $forr(x\cdot y)$ for $x,y\in\mathbb{R}^{2N}$.
\[ forr(x\cdot y) \triangleq \left\langle   \frac{1}{\sqrt{N}} H_N(x_1\cdot y_1) \Bigg| \frac{1}{\sqrt{N}} x_2\cdot y_2\right\rangle = \frac{1}{N\sqrt{N}}\sum_{i,j\in [N]} (-1)^{\langle i \mid j\rangle_{\mathbb{F}_2}} x_1(i)y_1(i)x_2(j)y_2(j) \] 
We make the following series of observations.
\begin{enumerate}[(1.)]
\item When $x$ and $y$ are drawn independently from $U_{2N}$, the random variable $forr(x\cdot y)$ has mean zero. Furthermore, it is highly concentrated around its mean. This can be seen as follows. The set $\{x_1(i),x_2(i),y_1(i), y_2(i)\}_{i\in [N]}$ is a set of independent $\{-1,1\}$-random variables with mean 0. This implies that the set of products $\{x_1(i)y_1(i)x_2(j)y_2(j)\}_{i, j\in [N]}$ is a set of pairwise independent $\{-1,1\}$-random variables with mean 0. Since $forr(x\cdot y)$ is a weighted sum of $N^2$ variables from this set, its variance can be computed to be at most $\frac{N^2}{(N\sqrt{N})^2}=\frac{1}{N}$. Let $A$ denote the event that $forr(x\cdot y)\le \epsilon/8$. Chebyshev's inequality implies that
\[ \underset{(x,y)\sim U_{4N}}{\p} [\neg A] \le \frac{64}{N\epsilon^2}\]
For $N$ greater than a sufficiently large constant, we have $\frac{64}{N\epsilon^2}\le \frac{\epsilon}{16}$. Thus,
\[ \underset{(x,y)\sim U_{4N}}{\p} [A] \ge 1- \frac{\epsilon}{16} \]
\item For every $x,y\in \{-1,1\}^{2N}$ and deterministic protocol $D\sim D_R$, since $D(x,y)\in\{-1,1\}$, we have 
\[\underset{(x,y)\sim U_{4N}|\neg A}{\E}[R(x,y)]\ge -1 \]
\item Whenever the event $A$ occurs, we have $forr(x\cdot y)\le \epsilon/8$, by definition of $A$. Hence, the distribution $U_{4N}|A$ is a distribution over \textsc{no} instances of the forrelation problem. This implies that
\[ \underset{(x,y)\sim U_{4N}|A}{\E}[R(x,y)]\ge 1-\frac{\epsilon}{16}  \]
\end{enumerate}
These observations allow us to conclude the following.
\begin{align*}\begin{split}
\underset{(x,y)\sim U_{4N}}{\E}[R(x,y)] &= \p[A] \cdot \underset{(x,y)\sim U_{4N}|A}{\E}[R(x,y)] + \p[\neg A]\cdot \underset{(x,y)\sim U_{4N}|\neg A}{\E}[R(x,y)] \\
&\ge \left(1 - \frac{\epsilon}{16}\right)\left(1-\frac{\epsilon}{16}\right) + \left(\frac{\epsilon}{16}\right)\times (-1)  \\
&\ge 1- \frac{3\epsilon}{16} 
\end{split}\end{align*} 

For simplicity of notation, let $V$ be the distribution on $\{-1,1\}^{2N}\times \{-1, 1\}^{2N}$ defined in Theorem \ref{theorem1}. This distribution is obtained by sampling $z\sim \mathcal{D}, x\sim U_{2N}$ and outputting $(x,x\cdot z)$. We make a series of observations analogous to the previous case.
\begin{enumerate}[(1.)]
\item For $(x,y)\sim V$, the distribution of $x\cdot y$ is $\mathcal{D}$. Lemma \ref{raztallemma} applied to $x\cdot y$ implies that 
\[ \underset{(x, y)\sim V}{\E}[forr(x\cdot y)] \ge \frac{\epsilon}{2} \]
Let $B$ denote the event that $forr(x\cdot y)\ge \frac{\epsilon}{4}$. Markov's inequality applied on the $[-1,1]$-random variable $forr(x\cdot y)$ implies that
\[ \underset{(x,y)\sim V}{\p} [B] \ge \frac{\epsilon}{4}  \]
\item For every $x,y\in \{-1,1\}^{2N}$ and deterministic protocol $D\sim D_R$, since $D(x,y)\in\{-1,1\}$, we have 
\[ \underset{(x,y)\sim V|\neg B}{\E}[R(x,y)] \le 1 \]
\item Whenever the event $B$ occurs, we have $forr(x\cdot y)\ge \epsilon/4$ by definition. Hence, the distribution $V|B$ is a distribution over \textsc{yes} instances of the forrelation problem. This implies that
\[ \label{eqnnum22} \underset{(x,y)\sim V|B}{\E}[R(x,y)] \le -1+ \frac{\epsilon}{16} \le0  \]
\end{enumerate}

These observations allow us to conclude the following.
\begin{align*}
\begin{split}
\underset{(x,y)\sim V}{\E}[R(x,y)] &= \p[B] \cdot \underset{(x,y)\sim V|B}{\E}[R(x,y)] + \p[\neg B]\cdot \underset{(x,y)\sim V|\neg B}{\E}R(x,y) \\
&\le 0 + \left(1-\frac{\epsilon}{4}\right)\times (+1) \\
&\le 1 - \frac{\epsilon}{4} 
\end{split}
\end{align*} 

These two conclusions imply that the protocol $R$ distinguishes $V$ and $U_{4N}$ with considerable advantage, that is,
 \begin{align*}\begin{split}
&E_{(x,y)\sim V}[R(x,y)] - \E_{(x,y)\sim U_{4N}}[R(x,y)] \\
&\ge 1-\frac{3\epsilon}{16}-\left(1-\frac{\epsilon}{4}\right) \\
&\ge \frac{\epsilon}{16} 
\end{split}\end{align*} 
Since $R(x,y)=\E_{D\sim D_R}[D(x,y)]$, this implies that 
\[ \E_{(x,y)\sim V} \E_{D\sim D_R}[D(x,y)] - \E_{(x,y)\sim U_{4N}} \E_{D\sim D_R}[D(x,y)] \ge \frac{\epsilon}{16}\]
Fix $D\sim D_R$ such that $ \E_{(x,y)\sim V}[D(x,y)] - \E_{(x,y)\sim U_{4N}} [D(x,y)]$ is at least the R.H.S. of the above. For this deterministic protocol $D$ of cost at most $O(c\ln\ln N)$, we have
 \[\E_{(x,y)\sim V}[D(x,y)] - \E_{(x,y)\sim U_{4N}}[D(x,y)] \ge \frac{\epsilon}{16} \] 
Theorem \ref{theorem1} applied to $D$ implies that $\frac{(c\ln\ln N)^2}{N^{1/2}}\ge \Omega(\epsilon)$. Since $\epsilon=\frac{1}{50\ln N}$, this implies that $c=\tilde{\Omega}(N^{1/4})$. This completes the proof of Theorem \ref{theorem3}.  \qed

\bibliographystyle{alpha}

\end{document}